\documentclass[12pt, amsfonts, epsfig, reqno, a4paper]{amsart}
\usepackage[hyperindex,breaklinks]{hyperref}
\usepackage{amsmath}
\usepackage{amsfonts}
\usepackage{mathrsfs}
\usepackage{graphicx}
\usepackage{color}
\usepackage{slashed}
\usepackage{braket}
\usepackage{extarrows}
\usepackage{amsmath, bm}
\usepackage{amssymb}
\usepackage[margin=1in]{geometry}
\usepackage{bbm}
\usepackage{arydshln}
\usepackage{hyperref}

\usepackage{breqn}

\newcommand{\cP}{{\mathcal P}}
\newcommand{\bZ}{{\mathbb Z}}
\newcommand{\Res}{{\mathrm{Res}}}

\newcommand{\half}{\frac12}

\definecolor{yellow}{rgb}{1,.7,0}

\definecolor{pkured}{rgb}{0.55,0,0}

\newcommand{\be}{\begin{equation*}}
	\newcommand{\ee}{\end{equation*}}
\newcommand{\beq}{\begin{equation}}
	\newcommand{\eeq}{\end{equation}}

\numberwithin{equation}{section}

\newtheorem{cor}{Corollary}[section]

\newtheorem{lem}[cor]{Lemma}
\newtheorem{prop}[cor]{Proposition}
\newtheorem{thm}[cor]{Theorem}

\newtheorem{ex}[cor]{Example}
\newtheorem{rmk}[cor]{Remark}

\numberwithin{figure}{section}
\usepackage{ytableau}

\usepackage{tikz}
\newcounter{x}
\newcounter{y}
\newcounter{z}

\author{Chenglang Yang}
\email{yangcl@whu.edu.cn}
\address{Institute for Math and AI, Wuhan University, Wuhan 430072, China}

\title[The $n$-Point Function of $t$-Core Partitions and Topological Vertex]
{The $n$-Point Function of $t$-Core Partitions and Topological Vertex}

\begin{document}
\maketitle

\begin{abstract}
	In this paper,
	we study the $n$-point function of $t$-core partitions.
	The main tool is the topological vertex, originally developed to study the topological string theory for toric Calabi--Yau 3-folds.
	By virtue of the topological vertex,
	we introduce the $q$-deformed $n$-point function
	that generalizes both the ordinary $n$-point function of all integer partitions studied by Bloch--Okounkov and $t$-core partition case treated here.
	As a consequence,
	we provide a closed formula for the $n$-point function of $t$-core partitions in terms of theta functions,
	and prove that the corresponding correlation functions are quasimodular forms.
\end{abstract}

\setcounter{section}{0}
\setcounter{tocdepth}{2}


\section{Introduction}

An integer partition $\lambda=(\lambda_1,\dots,\lambda_l)$ is a sequence of non-increasing positive integers.
This partition $\lambda$ is called a $t$-core partition if none of its hook lengths is a multiple of $t$.
It is known that the set of $t$-core partitions plays an important role in the modular representation theory of the symmetric groups.
Thus the $t$-core partitions have attracted considerable attention from mathematicians in the areas of combinatorics, representation theory, number theory and random partitions (see \cite{And08,CKNS,GKS,Han,JK81,KR14,Lam,WW20} and references therein).
The topological vertex, proposed by physicists \cite{AKMV},
is motivated by knot theory and Chern--Simons theory.
It is the building block for the topological string theory of toric Calabi--Yau 3-folds (see \cite{ADKMV,LLLZ,MOOP}).
In this paper,
we develop a novel approach to study the $n$-point function of $t$-core partitions using the topological vertex.

For the correlation functions of all integer partitions, 
motivated by mirror symmetry of elliptic curves,
they were proved to be quasimodular forms by Dijkgraaf \cite{Dij95} and Kaneko--Zagier \cite{KZ95} for special cases,
and by Bloch--Okounkov \cite{BO} for the general case.
There are some significant new proofs and generalizations of the above results (see \cite{CMZ18,CW04,Eng21,GM20,HIL,Mil03,TW07,W04,Z16,Zhou23}).
It is known that the set of all integer partitions not only labels a basis of the charge zero infinite wedge space,
but also represents all the irreducible representations of the symmetric groups.
Restricting to certain subsets of integer partitions often leads to refined structures that are rooted in deeper aspects of representation theory.
For instance, the set of self-conjugate partitions exhibits additional symmetry,
and its correlation functions were recently studied by Wang and the author using the $\omega$-transform on fermions \cite{WY24}.
As $t$-core partitions arise naturally from both modular representations of symmetric groups and representation theory of affine Lie algebras, a natural question is whether their correlation functions also enjoy the quasimodularity and other properties.
In this paper,
we answer this question by developing a topological vertex approach to the $n$-point functions of $t$-core partitions.
To be precise,
we study the following $n$-point function of $t$-core partitions
\begin{align}\label{eqn:main def F_t}
	F_t(Q;s_1,\cdots,s_n)
	=\frac{1}{\sum_{\nu\in\mathcal{P}_{\text{t-core}}} Q^{|\nu|}}
	\cdot \sum_{\nu\in\mathcal{P}_{\text{t-core}}} \bigg(\prod_{j=1}^n \sum_{i=1}^\infty s_j^{\nu_i-i+\half}
	\cdot Q^{|\nu|}\bigg),
\end{align}
where $\mathcal{P}_{\text{t-core}}$ denotes the set of all $t$-core partitions.
Despite the deceptive simplicity of their definition, the structure of $t$-core partitions is intricate.
As James and Kerber aptly noted in Section 2.7 of \cite{JK81}: "For $q\geq3$, the complete set of $q$-cores is rather complicated to describe".
Thus,
extracting the subset consisting of $t$-core partitions from the whole set of integer partitions is a non-trivial task.
For example,
the generating function for the number of $t$-core partitions,
which serves as the normalization constant in our $n$-point function,
was derived by \cite{GKS,Kly82,Ols} using the celebrated bijection that connects integer partitions and their $t$-cores.
The modularity of the generating function for the number of $t$-core partitions is intensively studied (see \cite{Gar93,GO96,KR14,Kly82}) since it can be expressed in terms of the Dedekind eta function.

We will introduce the $q$-deformed $n$-point function in terms of the topological vertex.
This function not only generalizes the ordinary $n$-point function of all integer partitions studied by Bloch and Okounkov \cite{BO},
but also gives the $n$-point function of $t$-core partitions under certain special limits.
As a consequence,
properties of the $q$-deformed $n$-point function can be used to study the $n$-point function of $t$-core partitions.
More precisely,
denote by $C_{\mu,\nu,\lambda}(q)$ the topological vertex,
which is a rational function in $q^{1/2}$ depending on three partitions $\mu,\nu,\lambda$.
Then the $q$-deformed $n$-point function $Z(Q;Q_1,q;s_1,\dots,s_n)$ is defined by
\begin{align}\label{eqn:def qZn main}
	\begin{split}
		\frac{1}{Z(Q;Q_1,q)}
		\cdot\sum_{\mu,\nu\in\mathcal{P}}
		\bigg((-Q_1)^{|\mu|}
		(-Q)^{|\nu|}
		C_{\emptyset,\mu^t,\nu}(q)
		C_{\emptyset,\mu,\nu^t}(q)
		\cdot \prod_{j=1}^n \sum_{i=1}^\infty s_j^{\nu_i-i+\half}\bigg),
	\end{split}
\end{align}
where $\mathcal{P}$ means the set of all integer partitions and the normalization constant $Z(Q;Q_1,q)$ is the $q$-deformed partition function defined in equation \eqref{eqn:def qZ}.
See subsection \ref{sec:qZn and Ft} for more details.
The first main result in this paper is the following relation between the $q$-deformed $n$-point function $Z(Q;Q_1,q;s_1,\dots,s_n)$ and the $n$-point function of $t$-core partitions $F_t(Q;s_1,\dots,s_n)$.
\begin{thm}[=Theorem \ref{prop:qZn lim Zn}]
	\label{thm:main qZn lim Zn}
	For any integer $t$ which is greater than $1$,
	denote by $\xi_t$ the $t$-th root of unity.
	Then we have
	\begin{align}\label{eqn:qZn lim Zn main}
		F_t(Q;s_1,\dots,s_n)
		=\lim_{q\rightarrow1^-} Z(Q;Q_1,q;s_1,\dots,s_n)|_{Q_1\rightarrow q^t, q\rightarrow \xi_t q}.
	\end{align}
\end{thm}

The $q$-deformed $n$-point function is motivated by the works \cite{BO,INRS,NO06,Y25}.
From the physical perspective, the partition function $Z(Q;Q_1,q)$ defined in equation \eqref{eqn:def qZ} can
be viewed as a closed topological string partition function of the local toric Calabi--Yau geometry, which is closely related to the Nekrasov partition function of 5D $\mathcal{N}=1$ gauge theory (see \cite{HIV,LLZ,NO06}).
The specialization used in equation \eqref{eqn:qZn lim Zn main} has appeared in \cite{Kim}.

The above formula \eqref{eqn:qZn lim Zn main} is very helpful to study the $n$-point function of $t$-core partitions.
Actually,
based on the cyclic symmetry of the topological vertex and a version of Wick Theorem,
we are able to obtain a formula for the $q$-deformed $n$-point function $Z(Q;Q_1,q;s_1,\dots,s_n)$,
which is expressed as a certain coefficient of a determinant (see equation \eqref{eqn:qZn as det}).
Then together with equation \eqref{eqn:qZn lim Zn main}, this provides a formula for the $n$-point function of $t$-core partitions $F_t(Q;s_1,\dots,s_n)$.
After further computations,
we can finally derive the following closed formula for $F_t(Q;s_1,\dots,s_n)$ in terms of theta functions.
\begin{thm}[=Theorem \ref{thm:Ft closed formula}]
	\label{thm:main formula}
	The $n$-point function of $t$-core partitions $F_t(Q;s_1,\dots,s_n)$ has the following closed formula
	\begin{align}\label{eqn:main Ft closed formula}
		\begin{split}
			F_t(Q;s_1,\dots,s_n)&
			=\frac{1}
			{\prod_{j=1}^n(s_j^{t/2}-s_j^{-t/2}) \cdot \Theta_3(-Q_2s_{[n]}^{-1})}
			\cdot\sum_{k=1}^{n}
			\bigg(\frac{1}{\Theta_3(-Q_2)^{k-1}}\\
			&\cdot
			\sum_{\{\pi_1,\dots,\pi_k\}\vdash[n]}
			\sum_{l_1,\dots,l_k=1}^t
			\prod_{m=1}^k \frac{\prod_{a=1}^t\vartheta(s_{\pi_m}\xi_t^{a-l_m})}
			{\prod_{a=1 \atop a\neq l_m}^t\vartheta(\xi_t^{a-l_m})}
			\cdot \det(A_{l_1,\dots,l_k;i,j}^{\pi_1,\dots,\pi_k})_{i,j=1}^{k}\bigg),
		\end{split}
	\end{align}
	where $[n]=\{1,2,...,n\}$,
	the second sum runs over all the length $k$ set partitions $\{\pi_1,\dots,\pi_k\}$ of $[n]$,
	$s_{\pi_m}=\prod\limits_{x\in\pi_m} s_{x}$,
	and the entries in the determinant are $A_{l_1,\dots,l_k;i,j}^{\pi_1,\dots,\pi_k}
	=\frac{\Theta_3(-Q_2s_{\pi_{i}}^{-1}\xi_t^{l_{i}-l_{j}})}
	{\vartheta(s_{\pi_{i}}\xi_t^{l_{j}-l_{i}})}$.
\end{thm}
Notice that
the $n$-point function of $t$-core partitions $F_t(Q;s_1,\dots,s_n)$ does not depend on the variable $Q_2$,
which shows that the right hand side of equation \eqref{eqn:main Ft closed formula} is independent of $Q_2$ even though this independence is far from obvious.
Consequently,
one can specialize $Q_2$ at special values,
which also give closed formulas for the $n$-point function of $t$-core partitions $F_t(Q;s_1,\dots,s_n)$ like those performed in Corollary \ref{cor:Ft closed formula-Q2}.
As an application,
we can prove the quasimodularity for the correlation functions of $t$-core partitions.
\begin{prop}[=Proposition \ref{prop:quasimod}]
	\label{thm:main quasimod}
	Let $s_j=e^{z_j}, j=1,\dots,n$ and $Q=e^{2\pi i \tau}$.
	We expand the $n$-point function $F_t(Q;s_1,\dots,s_n)$ of the $t$-core partitions by
	\begin{align*}
		F_t(Q;s_1,\dots,s_n)
		=\sum_{l_1,l_2,\dots,l_n=0}^\infty
		\langle f_{l_1}f_{l_2}\dots f_{l_n}\rangle_Q^{\textup{t-core}}
		\cdot \prod_{j=1}^n z_j^{l_j-1}.
	\end{align*}
	Then for any $l_1,\dots,l_n$,
	the correlation function of $t$-core partitions $\langle f_{l_1}f_{l_2}\dots f_{l_n}\rangle_Q^{\textup{t-core}}$ is a quasimodular form of weight at most $\sum_{j=1}^n l_j$ for the congruence subgroup $\Gamma_1(t)$.
\end{prop}

The quasimodularity for correlation functions in various cases is widely studied.
For example,
the foundational works \cite{BO,Dij95,Eng21,EO06,KZ95} prove the quasimodularity for correlation functions of all integer partitions
and the Hurwitz theory of elliptic orbifolds.
See \cite{CMZ18,GM20,HIL,Mil03,TW07,W04,WY24,Z16,Zhou23,Zhu96} and references therein for more details and generalizations.
However,
to the author's knowledge,
there are few cases that admit a direct closed-form formula as in equation \eqref{eqn:main Ft closed formula},
apart from the original one established by Bloch and Okounkov \cite{BO} for the $n$-point function of all integer partitions,
and some interesting generalizations in \cite{CW04,TW07} which rely on the Bloch--Okounkov's result.
We recommend the paper \cite{Z16} by Zagier for a new proof  of the Bloch--Okounkov's formula,
and \cite{Zhou23} by Zhou for a new formula of the Bloch--Okounkov's $n$-point function.

The rest of this paper is organized as follows.
In Section \ref{sec:pre},
we review the concepts of partitions, Schur functions, fermions and topological vertex.
In Section \ref{sec:qZn},
we introduce the $q$-deformed $n$-point function,
provide two formulas to compute it,
and prove Theorem \ref{thm:main qZn lim Zn}.
In Section \ref{sec:for Ft},
we derive the closed formula for the $n$-point function of $t$-core partitions,
which gives Theorem \ref{thm:main formula}.
As an application,
in Section \ref{sec:app}, we prove Proposition \ref{thm:main quasimod},
which provides the quasimodularity for the correlation functions of $t$-core partitions.

\section{Preliminaries}
\label{sec:pre}
In this section,
we review the concepts of integer partitions,
$t$-core partitions,
Schur functions,
fermionic Fock space and topological vertex.

\subsection{The integer partitions}

An integer partition is a sequence of weakly decreasing positive integers $\lambda=(\lambda_1,\dots,\lambda_l)$,
i.e., it satisfies $\lambda_1\geq\lambda_2\geq\dots\geq\lambda_l>0$ and $\lambda_i\in\mathbb{Z}$.
The number $l(\lambda)=l$ is called the length of this partition and $|\lambda|=\sum_{i=1}^l \lambda_i$ is called the size of this partition.
For convenience,
we will use $\emptyset$ or $(0)$ to represent the empty partition, both of whose length and size are equal to zero.
Denote by $\mathcal{P}$ the set of all integer partitions.

It is well-known that a partition $\lambda$ can be recorded as $\lambda=(1^{m_1}2^{m_2}\dots n^{m_n}\dots)$,
where $m_n:=\#\{i|\lambda_i=n\}$ is the number of times that $n$ appears within $\lambda=(\lambda_1,\dots,\lambda_l)$,
where $n$ is a positive integer.
As a consequence,
the generating function for the number of integer partitions is equal to
\begin{align}\label{eqn:size-generating function lambda}
	\sum_{\lambda\in\mathcal{P}} Q^{|\lambda|}
	=\sum_{m_1,m_2,\dots=0}^\infty Q^{\sum_{n=1}^\infty n\cdot m_n}
	=\frac{1}{\prod_{n=1}^\infty (1-Q^n)}.
\end{align}

Every integer partition $\lambda=(\lambda_1,\dots,\lambda_l)$ uniquely corresponds to its Young diagram,
which is a collection of boxes and has $\lambda_i$ boxes in its $i$-th row.
For example,
the left picture in Figure \ref{eqn:Yd 64421} provides the Young diagram of the partition $(6,4,4,2,1)$.
The conjugate of $\lambda$ is a new partition denoted by $\lambda^t$,
whose length is $\lambda_1$
and whose $k$-th part is $\lambda^t_k=\#\{j|\lambda_j\geq k\}, 1\leq k\leq \lambda_1$.
Intuitively,
the Young diagram of $\lambda^t$ is obtained from that of $\lambda$ by reflecting along the main diagonal.
For example, we have $(6,4,4,2,1)^t=(5,4,3,3,1,1)$,
whose Young diagram is the right picture in Figure \ref{eqn:Yd 64421}.
\begin{figure}[htbp]
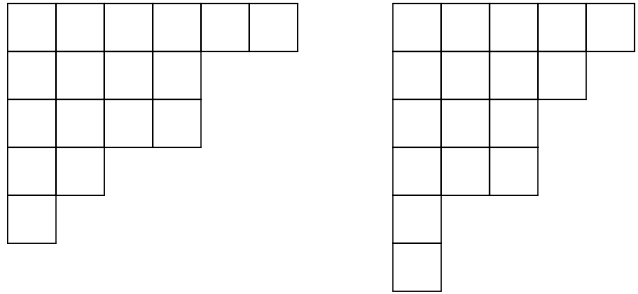

	\ydiagram{6,4,4,2,1}\qquad\quad
	\ydiagram{5,4,3,3,1,1}
	\caption{Young diagrams corresponding to $(6,4,4,2,1)$ and $(5,4,3,3,1,1)$}
	\label{eqn:Yd 64421}
\end{figure}

For any given partition $\lambda$,
there are many important combinatorial numbers coming from its Young diagram,
which is deeply related to representation theory and mathematical physics.
We use the notation $(j,k)\in\lambda$ to represent the $(j,k)$ box in the Young diagram corresponding to $\lambda$,
where $(j,k)$ should satisfy $1\leq j \leq l(\lambda)$, $1\leq k \leq \lambda_j$.
Then the content of this $(j,k)$ box is defined to be $c(j,k)=k-j$,
and the hook length is defined by $h(j,k)=\lambda_j+\lambda^t_k-j-k+1$.

\subsection{The $t$-core partitions}

A partition $\nu\in\mathcal{P}$ is called a $t$-core partition if none of its hook lengths is a multiple of $t$.
We denote by $\mathcal{P}_{\text{t-core}}$ the set of all the $t$-core partitions.
This notion plays an important role in the modular representation theory of symmetric groups and combinatorics (see \cite{And98,CKNS,GKS,Han,JK81,Ols} and references therein for more details).

The generating function for the number of $t$-core partitions is given by (see Proposition 3.3 in \cite{Ols}, Theorem 1.1 in \cite{Kly82}, equation (2.2) in \cite{GKS})
\begin{align}\label{eqn:size-generating function t-core lambda}
	\sum_{\lambda\in\mathcal{P}_{\text{t-core}}} Q^{|\lambda|}
	=\prod_{n=1}^\infty \frac{(1-Q^{nt})^n}{(1-Q^n)}.
\end{align}
We call it the partition function of $t$-core partitions.
For the $t=2$ case,
all the 2-core partitions are given by $(k,k-1,\dots,1), k\in\mathbb{N}$.
Then in this case equation \eqref{eqn:size-generating function t-core lambda} is due to Gauss (see, for example, Section 1 Ex. 91 in \cite{Stan1}).
In general,
it is not an easy task to distinguish $t$-core partitions from the set of all integer partitions (see \cite{JK81}),
so one cannot directly derive equation \eqref{eqn:size-generating function t-core lambda} as easily as we did for equation \eqref{eqn:size-generating function lambda}.
The method used in \cite{GKS,Kly82,Ols} to prove formula \eqref{eqn:size-generating function t-core lambda} all relies on the bijection between integer partitions and their $t$-cores and $t$-quotients. 
In Remark \ref{rmk:new pf Z},
we will give a new proof of formula \eqref{eqn:size-generating function t-core lambda} based on the $q$-deformed partition function.

\begin{rmk}
	For an arbitrary integer partition $\lambda$,
	the $t$-core of $\lambda$ is a new partition $\lambda^{\text{t-core}}$, which can be obtained by deleting all its rim hooks of size $t$ (see, for example Section 2.7 in \cite{JK81}).
	Thus, a general integer partition $\lambda$ is a $t$-core partition if and only if its $t$-core is equal to itself,
	i.e., $\lambda^{\text{t-core}}=\lambda$.
\end{rmk}

To better describe the integer partitions,
one can use the $0/1$-sequence method,
which is also known as the Maya diagram.
More precisely,
for a given partition $\lambda$,
the corresponding $0/1$-sequence is a sequence of $\{0,1\}$ labeled by integers.
We denote it by $s(\lambda)=(s_i)_{i\in\mathbb{Z}}$ and it is constructed by
\begin{align*}
	s_i=\begin{cases}
		0, & \text{if\ }i\in\{\lambda_j-j|j\geq1\},\\
		1, & \text{otherwise},
	\end{cases}
\end{align*}
where we use $\lambda_j=0$ if $j>l(\lambda)$.
Notice that
a $0/1$-sequence $(s_i)_{i\in\mathbb{Z}}$ comes from a partition if and only if $\#\{i<0|s_{i}=1\}
=\#\{i\geq0| s_{i}=0\}.$
For example,
the $0/1$-sequence corresponding to the partition $\lambda=(6,4,4,2,1)$ (see the left picture in Figure \ref{eqn:Yd 64421}) is
\begin{align*}
	s(\lambda)
	=(\cdots s_{-2}s_{-1} | s_0s_1s_2\cdots)
	=(\cdots0010101|10011011\cdots),
\end{align*}
where the vertical bar appears between $s_{-1}$ and $s_0$.
Under this description,
the set of boxes $(i_1,i_2)\in\lambda$ corresponds to the set of pairs $(s_{j_1},s_{j_2})$ such that $j_1<j_2,
s_{j_1}=1, s_{j_2}=0$.
Moreover,
the hook length of the box $(i_1,i_2)$ is equal to $j_2-j_1$,
i.e., we have the following equality between two multisets,
\begin{align*}
	\{h(i_1,i_2)|(i_1,i_2)\in\lambda\}
	=\{j_2-j_1|j_1<j_2,
	s_{j_1}=1, s_{j_2}=0\}.
\end{align*}
This equivalent description is convenient to study the $t$-core partitions.
For example,
the following well-known fact can be easily derived from this viewpoint.
\begin{lem}\label{lem:t-core h=t}
	A partition $\nu$ is a $t$-core partition if and only if
	there is no $(j,k)\in\nu$ such that $h(j,k)=t$.
\end{lem}
\begin{proof}
	The "if" part is trivial.
	We prove the "only if" part below.
	Assume that $\nu$ is not a $t$-core partition,
	then there exists at least one box $(i_1,i_2)\in\nu$ such that $h(i_1,i_2)$ is a multiple of $t$.
	As a consequence,
	if we denote the $0/1$-sequence corresponding to $\nu$ by $s(\nu)=(s_i)_{i\in\mathbb{Z}}$,
	then there exists a pair $(j_1,j_2)$ such that $j_1<j_2,
	s_{j_1}=1, s_{j_2}=0$,
	and in particular $j_2=j_1+Nt$ for some positive integer $N$.
	Consider the sequence
	\begin{align*}
		(s_{j_1}=1, s_{j_1+t}, s_{j_1+2t},\dots, s_{j_1+(N-1)t}, s_{j_2}=0).
	\end{align*}
	Then there must be a natural number $a<N$ such that
	$s_{j_1+at}=1, s_{j_1+(a+1)t}=0$.
	Consequently, the hook length of the box corresponding to the pair $(j_1+at,j_1+(a+1)t)$ is exactly equal to $t$.
\end{proof}

\subsection{Schur functions}
In this subsection,
we review the Schur functions and a special evaluation of them.
We mainly follow the notations in \cite{Mac}.

Denote by $\Lambda_{\mathbb{Q}}$ the ring of formal symmetric functions in infinitely many variables $({\bm x})=(x_1,x_2,\dots)$.
Then the complete symmetric functions $h_r({\bm x})\in\Lambda_{\mathbb{Q}}$ are defined by
\begin{align*}
	\sum_{r=0}^\infty h_r({\bm x}) \cdot z^r
	=\prod_{i=1}^\infty \frac{1}{1-x_i \cdot z}.
\end{align*}
The set $\{h_r({\bm x})\}_{r=1}^\infty$ forms an algebraic basis of the ring $\Lambda_{\mathbb{Q}}$.
The Schur functions $s_{\lambda}({\bm x})$ are symmetric functions labeled by integer partitions $\lambda\in\mathcal{P}$.
They have many equivalent definitions.
By virtue of the complete symmetric functions,
they can be defined by
\begin{align*}
	s_{\lambda}({\bm x})
	=\det\big(h_{\lambda_i-i+j}({\bm x})\big)_{i,j=1}^{l(\lambda)}.
\end{align*}
It is well-known that the set $\{s_\lambda({\bm x})\}_{\lambda\in\mathcal{P}}$ forms a linear basis of the space $\Lambda_{\mathbb{Q}}$.

The Schur functions play an important role in representation theory, combinatorics, geometry, mathematical physics, etc.
Special evaluations of Schur functions are also very useful.
For example,
$s_{\lambda}(1^m,0,.\dots)$ are related to dimensions of irreducible representations of the general linear group $\text{GL}(m,\mathbb{C})$.
In this paper,
we will need the following special evaluation (see I.3 Example 1 in \cite{Mac}),
\begin{align}\label{eqn:schur q^rho}
	s_{\lambda}(q^{\rho})
	=s_{\lambda}(q^{-1/2}, q^{-3/2}, q^{-5/2}, \dots)
	=q^{-n(\lambda)-|\lambda|/2}\prod_{(i,j)\in\lambda} \frac{1}{1-q^{-h(i,j)}},
\end{align}
where $n(\lambda)=\sum_{i=1}^{l(\lambda)}(i-1)\lambda_i$,
$\rho=(-1/2,-3/2,-5/2,\dots)$,
so the notation $s_{\lambda}(q^{\rho})$ means the evaluation of $s_{\lambda}({\bm x})$ at the point $x_1=q^{-1/2}, x_2=q^{-3/2}, x_3=q^{-5/2}, \dots$.
This special evaluation is related to the quantum dimension and is widely used in the topological string theory of toric Calabi--Yau 3-fold geometry.
Notice that even though the Schur function $s_{\lambda}({\bm x})$ is a formal function in infinitely many variables,
the result of this special evaluation performed in equation \eqref{eqn:schur q^rho} is a rational function in $q^{1/2}$.

The skew Schur functions are generalizations of Schur functions.
Each of them is still a symmetric function but depends on two partitions.
To be precise,
for any two partitions $\lambda,\mu\in\mathcal{P}$,
the corresponding skew Schur functions are denoted by $s_{\lambda/\mu}({\bm x})$.
Then it is defined by
\begin{align*}
	s_{\lambda/\mu}({\bm x})
	=\sum_{\nu\in\mathcal{P}} c^{\lambda}_{\mu \nu} s_{\nu}({\bm x}),
\end{align*}
where the Littlewood--Richardson coefficients $\{c^{\lambda}_{\mu \nu}\}$ are given by $s_{\mu}({\bm x}) s_{\nu}({\bm x})=\sum_{\lambda} c^{\lambda}_{\mu \nu} s_{\lambda}({\bm x})$.
Moreover,
when $\mu=\emptyset$,
one has $s_{\lambda/\emptyset}({\bm x})=s_{\lambda}({\bm x})$.

\subsection{Fermions and fermionic Fock space}
We review the construction of the fermionic Fock space and fermionic operators.
We mainly follow the notation in \cite{DJM}.
Let $\bm{a}=(a_1,a_2,\dots)$ be a decreasing sequence of half-integers $a_i\in\mathbb{Z}+\half$.
It is called admissible if both of the following sets
\begin{align}\label{eqn:def S+S-}
	S_-=(\mathbb{Z}_{\leq 0}+\half)\setminus\{a_i|i\in\mathbb{Z}_{\geq1}\},
	\qquad\text{and}
	\qquad S_+=\{a_i|i\in\mathbb{Z}_{\geq1}\}\setminus(\mathbb{Z}_{\leq 0}-\half)
\end{align}
are finite.
For each admissible sequence $\bm{a}=(a_1,a_2,\dots)$,
we assign a vector
\begin{align*}
	|\bm{a}\rangle
	=\underline{a_1}\wedge\underline{a_2}\wedge\underline{a_3}\wedge\cdots.
\end{align*}
In particular,
the vector associated with $(-\half,-\frac{3}{2},-\frac{5}{2}, -\frac{7}{2},\dots)$ is called the vacuum,
which is denoted by $|0\rangle=\underline{-\frac{1}{2}}\wedge\underline{-\frac{3}{2}}\wedge\underline{-\frac{5}{2}}\wedge\cdots$.
Then the fermionic Fock space $\mathcal{F}$ is the linear space generated by all those vectors $|\bm{a}\rangle$ over $\mathbb{C}$.
More precisely,
each element in $\mathcal{F}$ is of the following form
\begin{align*}
	\sum_{\bm{a}\ \text{admissible}}
	c_{\bm{a}} \cdot |\bm{a}\rangle,
\end{align*}
where the summation is taken over all admissible sequences $\bm{a}$ and it could be an infinite sum.
There is a standard inner product on the fermionic Fock space $\mathcal{F}$ such that $\{|\bm{a}\rangle\}$ is an orthonormal basis.
This inner product is denoted by $(\cdot,\cdot)$.

The fermions are two families of operators acting on the fermionic Fock space $\mathcal{F}$.
They are denoted by $\{\psi_k\}_{k\in\mathbb{Z}+\half}$ and $\{\psi^*_k\}_{k\in\mathbb{Z}+\half}$.
The action of $\psi_k$ on a vector $|\mathbf{a}\rangle$ is given by
\begin{align*}
	\psi_k \cdot |\bm{a}\rangle = \underline{k}\wedge |\bm{a}\rangle.
\end{align*}
The $\psi_k^*$ is the adjoint operator of $\psi_k$,
whose action on $|\bm{a}\rangle$ is given by
\begin{align*}
	\psi^*_k \cdot |\bm{a}\rangle
	=\begin{cases}
		(-1)^{i+1} \underline{a_1}\wedge\underline{a_2}\wedge\cdots\wedge\widehat{\underline{a_i}}\wedge\cdots, & \text{\ if\ }a_i=k\text{\ for\ some\ }i;\\
		0, &\text{otherwise}.
	\end{cases}
\end{align*}
These two families of operators satisfy the following anti-commutation relations
\begin{align}\label{eqn:anti comm}
	[\psi_{k_1},\psi_{k_2}]_+=0,
	\qquad [\psi^*_{k_1},\psi^*_{k_2}]_+=0,
	\qquad [\psi_{k_1},\psi^*_{k_2}]_+=\delta_{k_1,k_2}\cdot \text{id},
\end{align}
where the bracket is defined by $[A,B]_+=AB+BA$.
There is a natural grading on the fermionic Fock space $\mathcal{F}$ which is called the charge.
For any $|\bm{a}\rangle$,
its charge is
\begin{align*}
	\text{charge}(|\bm{a}\rangle)=(|S_+|-|S_-|) \cdot |\bm{a}\rangle,
\end{align*}
where the two sets $S_+$ and $S_-$ are defined in equation \eqref{eqn:def S+S-}.
The operator $\psi_k$ has charge $+1$ and $\psi^*_k$ has charge $-1$.
We denote by $\mathcal{F}^{(0)}$ the charge zero fermionic Fock space.

We use $\mathcal{F}^*$ to represent the dual space of $\mathcal{F}$.
For a vector $|v\rangle\in\mathcal{F}$,
we use $\langle v|$ to represent its dual vector in $\mathcal{F}$.
Then the vacuum expectation value is given by
\begin{align*}
	\langle v_1| A |v_2\rangle
	:=(|v_1\rangle, A |v_2\rangle)
	=(A^* |v_1\rangle, |v_2\rangle),
\end{align*}
where $A$ is a linear operator acting on the fermionic Fock space and $A^*$ is its adjoint.

For each partition $\lambda=(\lambda_1,\lambda_2,\dots)$,
we can define an admissible sequence by
$(\lambda_1-\half, \lambda_2-\frac{3}{2},\dots,\lambda_i-i+\frac{1}{2},\dots),$
where we set $\lambda_i=0$ if $i>l(\lambda)$.
Thus the partition $\lambda$ corresponds to a vector in the fermionic Fock space which is denoted as
\begin{align}
	|\lambda\rangle
	=\underline{\lambda_1-\half}\wedge\underline{\lambda_2-\frac{3}{2}}\wedge\cdots.
\end{align}
Notice that the empty partition $\emptyset$ exactly corresponds to the vacuum vector $|0\rangle$. 
One can easily prove that $\{|\lambda\rangle\}_{\lambda\in\mathcal{P}}$ forms a basis of the charge zero fermionic Fock space $\mathcal{F}^{(0)}$.
The action of fermions $\psi_k$ and $\psi^*_k$ on $|\lambda\rangle$ is direct.
The consequence is
\begin{align}\label{eqn:psipsi* action}
	\psi_k\psi^*_k |\lambda\rangle
	=\begin{cases}
		|\lambda\rangle, &\text{if}\ k=\lambda_i-i+\half \text{\ for\ some}\ i,\\
		0, &\text{otherwise}.
	\end{cases}
\end{align}

To study the product of several fermionic operators,
one always uses the normal ordering product.
In particular,
for the special case of product of two operators,
we have
\begin{align*}
	:\psi_{k_1}\psi_{k_2}:
	=\psi_{k_1}\psi_{k_2},
	\ \ \ :\psi^*_{k_1}\psi^*_{k_2}:
	=\psi^*_{k_1}\psi^*_{k_2},
\end{align*}
for any $k_1, k_2$ and
\begin{align*}
	:\psi_{a}\psi^*_{b}:
	=-:\psi^*_{b}\psi_{a}:
	=\begin{cases}
		\psi_{a}\psi^*_{b}, & \text{if\ }b>0,\\
		-\psi^*_{b}\psi_{a}, & \text{otherwise}.
	\end{cases}
\end{align*}
One can verify $:\phi_{k_1}\phi_{k_2}:=\phi_{k_1}\phi_{k_2}-\langle0|\phi_{k_1}\phi_{k_2}|0\rangle$,
where $\phi_k$ denotes a fermion, which can be either $\psi_k$ or $\psi^*_k$.
The fermionic fields are generating functions of fermions as
\begin{align*}
	\psi(z)=
	\sum_{k\in\mathbb{Z}+\half} \psi_k z^{k},
	\qquad\psi^*(w)=
	\sum_{k\in\mathbb{Z}+\half} \psi^*_k w^{-k}.
\end{align*}
The anti-commutation relations \eqref{eqn:anti comm} are equivalent to
\begin{align*}
	\psi(z)\psi(w)
	=-\psi(w)\psi(z),
	\  \ \ \psi^*(z)\psi^*(w)
	=-\psi^*(w)\psi^*(z),
\end{align*}
and
\begin{align*}
	\psi(z)\psi^*(w)
	=-\psi^*(w)\psi(z)
	+\sum_{k\in\mathbb{Z}+\half} z^k w^{-k}.
\end{align*}
Then,
the normal ordering product of $\psi(z)$ and $\psi^*(w)$ can be represented by the following
\begin{align*}
	:\psi(z)\psi^*(w):
	=\psi(z)\psi^*(w)-\sum_{k\in\mathbb{Z}_{\leq0}+\half} z^k w^{-k}.
\end{align*}

The translation operator $R$ is defined by
\begin{align}\label{eqn:def R}
	R\ \underline{a_1}\wedge\underline{a_2}\wedge\underline{a_3}\wedge\dots
	=\underline{a_1+1}\wedge\underline{a_2+1}\wedge\underline{a_3+1}\wedge\cdots.
\end{align}
Then $R$ is of charge 1 as an operator.
From the definitions of $\psi_k$ and $\psi^*_k$,
we have $R \psi_k R^{-1}=\psi_{k+1}$ and $R \psi^*_k R^{-1}=\psi^*_{k+1}$.
Correspondingly,
at the level of fermionic fields,
we have
\begin{align*}
	R^N \psi(z) R^{-N}=z^{-N}\psi(z)
	\qquad\text{and}\qquad
	R^N \psi^*(z) R^{-N}=z^N \psi^*(z)
\end{align*}
for any integer $N$.

\subsection{Vertex operators}
In this subsection,
we introduce the vertex operators.
Let us start with the bosons,
which are some operators acting on the fermionic Fock space $\mathcal{F}$ defined by
\begin{align}\label{eqn:def alphan}
	\alpha_n = \sum_{k\in\mathbb{Z}+\half} :\psi_{k} \psi^*_{k+n}:,
	\ n\in\mathbb{Z}.
\end{align}
They satisfy the following commutation relation
\begin{align*}
	[\alpha_m, \alpha_n]=m\delta_{m,-n} \cdot \text{Id}.
\end{align*}
Notice that the dual operator of $\alpha_n$ is $\alpha_{-n}$. 
Then the vertex operators are generating functions of these bosons defined by
\begin{align*}
	\Gamma_\pm(z)
	=\exp\Big(\sum_{b=1}^\infty \frac{z^b}{b} \alpha_{\pm b}\Big).
\end{align*}
The vertex operators satisfy the following commutation relations
\begin{align}\label{eqn:comm Gammapm}
	[\Gamma_\pm(z),\Gamma_{\pm}(w)]=0,
	\qquad [\Gamma_\pm(z),\Gamma_{\mp}(w)]=\frac{1}{(1-zw)^{\pm1}}
\end{align}
with the assumption $|zw|<1$.
For infinitely many variables $\bm x=(x_1,x_2,\dots)$ and $\bm y=(y_1,y_2,\dots)$,
we usually use the following notation
\begin{align}
	\Gamma_\pm(\bm x)
	=\prod_{i=1}^\infty \Gamma_\pm(x_i).
\end{align}
Then the commutation relation \eqref{eqn:comm Gammapm} gives
\begin{align}\label{eqn:comm Gammapminf}
	[\Gamma_{\pm}(\bm x),\Gamma_{\pm}(\bm y)]=0,
	\qquad [\Gamma_\pm(\bm x),\Gamma_{\mp}(\bm y)] =\prod_{j,k=1}^\infty \frac{1}{(1-x_jy_k)^{\pm1}}.
\end{align}

Equation \eqref{eqn:def alphan} tells us how to construct the bosons from fermions.
In the opposite direction,
one can reconstruct fermionic fields by the following formulas (see, for examples, Section 5 in \cite{DJM} and Section 14 in \cite{Kac})
\begin{align}\label{eqn:psi as gamma}
	\psi(z)
	=z^{C-\half}
	R \Gamma_-(z) \Gamma_+(z^{-1})^{-1}
	\quad\text{\ and\ }\quad
	\psi^*(z)
	=R^{-1}z^{-C+\half}
	\Gamma_-(z)^{-1} \Gamma_+(z^{-1}).
\end{align}
The famous boson-fermion correspondence deeply relies on this construction and reconstruction.
Using the vertex operators and vacuum expectation value,
it not only establishes a relation between the fermionic Fock space and the ring of formal symmetric functions,
but also tells us how to compute the actions of fermions and bosons.
In particular,
the vector $|\lambda\rangle\in\mathcal{F}^{(0)}$ corresponds to the Schur function $s_{\lambda}(\bm x)$.
To be precise,
we have
\begin{align}
	s_\lambda(\bm x)
	=\langle0|\Gamma_+(\bm x)|\lambda\rangle
	=\langle\lambda|\Gamma_-(\bm x)|0\rangle.
\end{align}
Moreover,
the skew Schur functions can also be represented in terms of vacuum expectation value as the following
\begin{align}\label{eqn:skew as vev}
	s_{\lambda/\mu}(\bm x)
	=\langle\mu|\Gamma_+(\bm x)|\lambda\rangle
	=\langle\lambda|\Gamma_-(\bm x)|\mu\rangle.
\end{align}

In this paper,
we use two techniques to obtain identities of vacuum expectation values involving vertex operators.
First,
by taking dual of vertex operators and vectors in the fermionic Fock space,
we have
\begin{align}\label{eqn:taking dual}
	\langle\mu|\Gamma_-(\bm{x}) \Gamma_+(\bm{y})|\nu\rangle
	=\langle\nu|\Gamma_-(\bm{y}) \Gamma_+(\bm{x})|\mu\rangle.
\end{align}
Second,
it was Wang and the author who introduced the $\omega$-transform on the fermionic Fock space and fermions (see \cite{WY24,Y25}),
which is a generalization of the usual $\omega$-transform on the ring of symmetric functions (see Section I.2 in \cite{Mac}).
In particular,
we have
\begin{align}\label{eqn:taking w}
	\langle\mu|\Gamma_-(\bm{x}) \Gamma_+(\bm{y})|\nu\rangle
	=\langle\mu^t|\Gamma_-(-\bm{x})^{-1} \Gamma_+(-\bm{y})^{-1}|\nu^t\rangle
\end{align}

There are two basic operators on the fermionic Fock space. 
The first is $C=\alpha_0= \sum_{k\in\mathbb{Z}+\half} :\psi_{k} \psi^*_{k}:$.
It is called the charge operator and satisfies
\begin{align*}
	C\ |\bm{a}\rangle
	=(|S_+|-|S_-|) \cdot |\bm{a}\rangle
	=\text{charge}(|\bm{a}\rangle) \cdot |\bm{a}\rangle.
\end{align*}
The second operator is
$L_0= \sum_{k\in\mathbb{Z}+\half} k :\psi_{k} \psi^*_{k}:$.
It is called the energy operator since
\begin{align}\label{eqn:action L0}
	L_0\ |\lambda\rangle
	=|\lambda| \cdot |\lambda\rangle.
\end{align}
By definition of the translation operator $R$ and direct computations,
for any integer $N$,
we have the following conjugate relations,
\begin{align}\label{eqn:comm RCH}
	R^{N} C R^{-N}=C-N,
	\ \ \ 
	R^{N} L_0 R^{-N}=L_0 -N C +N^2/2
\end{align}
and
\begin{align}\label{eqn:comm RGamma}
	R^{N} \Gamma_\pm(z) R^{-N}
	=\Gamma_\pm(z).
\end{align}

\subsection{Topological vertex}
\label{sec:top ver}
In this subsection,
we review the topological vertex and its properties.
It is one of the main tools used in this paper to study the $t$-core partitions.

Motivated by the duality between topological string theory and gauge theory,
Aganagic, Klemm, Mari\~{n}o and Vafa introduced the topological vertex to capture the A-model topological string amplitudes of smooth toric Calabi--Yau 3-folds.
The combinatorial expression of the topological vertex is given by
\begin{align}\label{eqn:def C}
	C_{\lambda,\mu,\nu}(q)
	= q^{\kappa(\lambda)/2+\kappa(\nu)/2}
	\cdot s_{\nu^t} (q^{\rho})
	\cdot \sum_{\eta\in\mathcal{P}} s_{\lambda^t /\eta} (q^{\nu + \rho})
	s_{\mu /\eta} (q^{\nu^t + \rho}),
\end{align}
where the notations have been reviewed in previous subsections.
This expression for the topological vertex comes from the colored HOMFLY-PT polynomials from knot theory and Chern--Simons theory \cite{ADKMV,AKMV}.

The mathematical theory of the topological vertex was established in terms of relative Gromov--Witten theory \cite{LLLZ},
and Donaldson--Thomas theory \cite{MOOP}.
In principle,
the topological vertex $C_{\lambda,\mu,\nu}(q)$ is a generating function of open Gromov--Witten invariants of $\mathbb{C}^3$.
Together with the gluing rules \cite{ADKMV,AKMV,LLLZ,MOOP},
it can be used to compute all (open) Gromov--Witten invariants of general smooth toric Calabi--Yau 3-folds.

The topological vertex has been intensively studied in the literature.
See \cite{ADKMV,AKMV,INRS,LLLZ,MOOP,ORV,WYZ24} and references therein.
It has beautiful combinatorial structures and possesses deep connections with many other fields in mathematics.
For example,
the topological vertex has the following symmetry of rotation,
\begin{align}\label{eqn:rotation symm}
	C_{\lambda,\mu,\nu}(q)
	=C_{\mu,\nu,\lambda}(q)
	=C_{\nu,\lambda,\mu}(q),
\end{align}
and symmetry of mirror,
\begin{align}
	q^{\|\lambda^t\|^2/2+\|\mu^t\|^2/2+\|\nu^t\|^2/2}
	C_{\lambda,\mu,\nu}(q)
	=q^{\|\lambda\|^2/2+\|\mu\|^2/2+\|\nu\|^2/2}
	C_{\mu^t,\lambda^t,\nu^t}(q),
\end{align}
where $\|\lambda\|=\sum_{i=1}^{l(\lambda)} \lambda_i^2$.
Moreover,
there are at least two methods to represent the topological vertex as the vacuum expectation value.
The first method is given by Okounkov, Reshetikhin and Vafa \cite{ORV}.
They first express the generating function of certain plane partitions as a vacuum expectation value involving vertex operators (see their equation (3.17)),
and then they show the equivalence between this generating function and topological vertex up to a MacMahon factor (see their equation (3.21)).
The second method, proposed by Wang, Zhou and the author,
expresses the topological vertex as a vacuum expectation value involving not only vertex operators but also fermionic fields (see Definition 4.2 and Theorem 4.2 in \cite{WYZ24}) without any additional factor.
This is the main tool used in \cite{WYZ24,WYZ25} to prove the multi-component KP integrability of open Gromov--Witten theory of general smooth toric Calabi--Yau 3-folds.
In this paper,
we only need a special case of the above mentioned vacuum expectation formula as
\begin{align}\label{eqn:clm0=<>}
	C_{\lambda,\mu,\emptyset}(q)
	=q^{\kappa(\lambda)/2}
	\cdot\sum_{\eta\in\mathcal{P}} s_{\lambda^t /\eta} (q^{\rho})
	s_{\mu /\eta} (q^{\rho})
	=q^{\kappa(\lambda)/2}
	\langle\lambda^t|
	\Gamma_-(q^\rho)
	\Gamma_+(q^\rho)|\mu\rangle.
\end{align}
The first equal sign comes from the combinatorial expression \eqref{eqn:def C} of the topological vertex,
and the second equal sign can be directly obtained from the property of skew Schur functions performed in equation \eqref{eqn:skew as vev}.

\section{The $q$-deformed $n$-point function and its relation with $t$-core partitions}
\label{sec:qZn}
In this section,
we introduce the $q$-deformed $n$-point function and establish its relation with the $n$-point function of $t$-core partitions.
The main tools are the topological vertex and free fermions.
This function is not only a generalization of the ordinary $n$-point function of all integer partitions studied by Bloch and Okounkov \cite{BO} (see Proposition \ref{prop:qZn lim F}),
but also a generalization of the $n$-point function of $t$-core partitions studied in this paper (see Theorem \ref{prop:qZn lim Zn}).

\subsection{The $q$-deformed $n$-point function}
\label{sec:qZn and Ft}
In this subsection,
we use the topological vertex to define the $q$-deformed $n$-point function,
which is denoted by $Z(Q;Q_1,q;s_1,\dots,s_n)$
and depends on two extra variables $Q_1, q$ compared to the $n$-point function of $t$-core partitions $F_t(Q;s_1,\dots,s_n)$.
We first consider the $q$-deformed partition function defined as
\begin{align}\label{eqn:def qZ}
	Z(Q;Q_1,q):=
	\sum_{\mu,\nu\in\mathcal{P}}
	(-Q_1)^{|\mu|}
	(-Q)^{|\nu|}
	C_{\emptyset,\mu^t,\nu}(q)
	C_{\emptyset,\mu,\nu^t}(q),
\end{align}
where $C_{\emptyset,\mu,\nu}(q)$ is a special case of the topological vertex introduced in subsection \ref{sec:top ver}.
As we will show in equation \eqref{eqn:t-core=Zq lim},
the $q$-deformed partition function $Z(Q;Q_1,q)$ can recover the ordinary generating function of $t$-core partitions by taking certain limit,
which is our motivation to consider the $q$-deformed $n$-point function.
We define it by
\begin{align}\label{eqn:def qZn}
	\begin{split}
	Z(Q;Q_1,&q;s_1,\dots,s_n)\\
	:=&\frac{1}{Z(Q;Q_1,q)}
	\cdot\sum_{\mu,\nu\in\mathcal{P}}
	\bigg((-Q_1)^{|\mu|}
	(-Q)^{|\nu|}
	C_{\emptyset,\mu^t,\nu}(q)
	C_{\emptyset,\mu,\nu^t}(q)
	\cdot \prod_{j=1}^n \sum_{i=1}^\infty s_j^{\nu_i-i+\half}\bigg).
	\end{split}
\end{align}

The main result of this subsection is the following relation between the $n$-point function of $t$-core partitions and the $q$-deformed $n$-point function.
\begin{thm}\label{prop:qZn lim Zn}
	For any integer $t$ which is greater than $1$,
	denote by $\xi_t$ the $t$-th root of unity.
	Then we have
	\begin{align}\label{eqn:qZn lim Zn}
		F_t(Q;s_1,\dots,s_n)
		=\lim_{q\rightarrow1^-} Z(Q;q^t,\xi_t q;s_1,\dots,s_n),
	\end{align}
	where we recall that
	$F_t(s_1,\dots,s_n)$ is the $n$-point function of $t$-core partitions as
	\begin{align}\label{eqn:def F_t}
		F_t(Q;s_1,\cdots,s_n)
		=\frac{1}{\sum\limits_{\nu\in\mathcal{P}_{\text{t-core}}} Q^{|\nu|}}
		\cdot \sum_{\nu\in\mathcal{P}_{\text{t-core}}} \bigg(\prod_{j=1}^n \sum_{i=1}^\infty s_j^{\nu_i-i+\half}
		\cdot Q^{|\nu|}\bigg).
	\end{align}
\end{thm}
\begin{proof}
Let us start with a much easier case.
Comparing equations \eqref{eqn:def qZn} and \eqref{eqn:def F_t},
we first prove that the normalization constants in these two functions are related to each other.
To be more precise,
we first prove the following
\begin{align}\label{eqn:t-core=Zq lim}
	\sum_{\nu\in\mathcal{P}_{\text{t-core}}} Q^{|\nu|}
	=\lim_{q\rightarrow1^-}
	\frac{Z(Q;Q_1,q)}{\prod_{j,k=1}^\infty (1-Q_1 q^{-j-k+1})}|_{q\rightarrow\xi_tq, Q_1\rightarrow q^t}.
\end{align}
Note that the notation $\lim\limits_{q\rightarrow1^-}f|_{q\rightarrow\xi_tq, Q_1\rightarrow q^t}$ means,
for the function $f$
we first replace the variables $q$ and $Q_1$ by $\xi_tq$ and $q^t$, respectively,
and after that we take the limit $q\rightarrow1^-$.
The order between change of the variables $q$ and $Q_1$ does not matter since we have $(\xi_t q)^t=q^t$.

By the definition of the topological vertex in equation \eqref{eqn:def C},
the $q$-deformed partition function defined in equation \eqref{eqn:def qZ} can be represented by
\begin{align}\label{eqn:Zq as ssss}
	Z(Q;Q_1,q)
	=&\sum_{\mu,\nu\in\mathcal{P}}
	(-Q_1)^{|\mu|}
	(-Q)^{|\nu|}
	s_{\nu^t}(q^{\rho})s_{\mu^t}(q^{\nu^t+\rho})
	s_{\nu}(q^{\rho})s_{\mu}(q^{\nu+\rho}).
\end{align}
To compute the right hand side of above equation,
we need to apply Lemma \ref{lem:from Y23} for the case of $\nu^1=\nu^2=\nu$ to take the summation over $\mu$,
and note that for $(j,k)\in\nu$, we have $h(j,k)=\nu_j+\nu^t_k-j-k+1$.
Then equation \eqref{eqn:Zq as ssss} reduces to
\begin{align*}
	Z(Q;Q_1,q)
	=\prod_{j,k=1}^\infty& (1-Q_1 q^{-j-k+1})\\
	&\cdot \sum_{\nu\in\mathcal{P}}
	(-Q)^{|\nu|}
	s_{\nu^t}(q^{\rho})s_{\nu}(q^{\rho})
	\prod_{(j,k)\in\nu}(1-Q_1q^{h(j,k)})(1-Q_1q^{-h(j,k)}).
\end{align*}
Thus,
by the evaluation formula \eqref{eqn:schur q^rho} for Schur function $s_{\nu}(q^{\rho})$,
we have
\begin{align}\label{eqn:Zq as 1}
	\begin{split}
	\frac{Z(Q;Q_1,q)}{\prod_{j,k=1}^\infty (1-Q_1 q^{-j-k+1})}
	=&\sum_{\nu\in\mathcal{P}}
	(-Q)^{|\nu|}
	q^{-n(\nu)-n(\nu^t)-|\nu|}
	\prod_{(j,k)\in\nu}\frac{(1-Q_1q^{h(j,k)})(1-Q_1q^{-h(j,k)})} {(1-q^{-h(j,k)})^2}\\
	=&\sum_{\nu\in\mathcal{P}}
	(QQ_1)^{|\nu|}
	\prod_{(j,k)\in\nu}\frac{(1-Q_1q^{h(j,k)})(1-Q_1^{-1}q^{h(j,k)})} {(1-q^{h(j,k)})^2},
	\end{split}
\end{align}
where we have used the combinatorial identity
$\sum_{(j,k)\in\nu} h(j,k)=n(\nu)+n(\nu^t)+|\nu|$,
which is proved in Lemma \ref{lem:comb lem for nu}.
We take $\xi_t$ to be the $t$-th root of unity and apply the change of variables $q\rightarrow \xi_t q, Q\rightarrow q^t$ to equation \eqref{eqn:Zq as 1}.
The result is
\begin{align}\label{eqn:qZ|cov=}
	\begin{split}
	&\frac{Z(Q;Q_1,q)}{\prod_{j,k=1}^\infty (1-Q_1 q^{-j-k+1})}|_{q\rightarrow\xi_tq, Q_1\rightarrow q^t}\\
	&\qquad\qquad\qquad=\sum_{\nu\in\mathcal{P}}
	(q^tQ)^{|\nu|}
	\prod_{(j,k)\in\nu}\frac{(1-\xi_t^{h(j,k)}q^{h(j,k)+t}) (1-\xi_t^{h(j,k)}q^{h(j,k)-t})}
	{\big(1-\xi_t^{h(j,k)}q^{h(j,k)}\big)^2}.
	\end{split}
\end{align}

We are now ready to take the limit $q\rightarrow1^-$.
For each given $(j,k)\in\nu$,
we examine the corresponding term in the right hand side of equation above in detail.
The limit $q\rightarrow1^-$ gives
\begin{align}\label{eqn:dis lim}
	\frac{(1-\xi_t^{h(j,k)}q^{h(j,k)+t}) (1-\xi_t^{h(j,k)}q^{h(j,k)-t})}
	{\big(1-\xi_t^{h(j,k)}q^{h(j,k)}\big)^2}
	\rightarrow\begin{cases}
		\frac{(h(j,k)+t) (h(j,k)-t)}
		{h(j,k)^2}, &\text{if\ } t| h(j,k),\\
		1, &\text{if\ } t\nmid h(j,k),
	\end{cases}
\end{align}
where in the case of $t|h(j,k)$,
we have used the L'Hospital's rule.
As a consequence,
after taking $\lim_{q\rightarrow1^-}$ in both sides of equation \eqref{eqn:qZ|cov=},
we have proved
\begin{align}\label{eqn:lim Zq=nu..}
	\lim_{q\rightarrow1^-}
	\frac{Z(Q;Q_1,q)}{\prod\limits_{j,k=1}^\infty (1-Q_1 q^{-j-k+1})}|_{q\rightarrow\xi_tq \atop Q_1\rightarrow q^t}
	=\sum_{\nu\in\mathcal{P}}
	\bigg(\prod_{(j,k)\in \nu \atop t|h(j,k)} \frac{\big(h(j,k)+t\big) \big(h(j,k)-t\big)}{h(j,k)^2}
	\cdot Q^{|\nu|}\bigg).
\end{align}
By Lemma \ref{lem:t-core h=t},
if $\nu$ is not a $t$-core partition,
there must be at least one box $(j,k)\in\nu$ such that $h(j,k)=t$.
Thus the corresponding contribution of $\nu$ in the right hand side of equation \eqref{eqn:lim Zq=nu..} vanishes due to the factor $\big(h(j,k)-t\big)$.
Otherwise, if $\nu$ is a $t$-core partition,
then by definition there is no $(j,k)\in\nu$ such that $t|h(j,k)$.
Thus the corresponding contribution of $\nu$ will simply be $Q^{|\nu|}$.
Consequently,
the above equation \eqref{eqn:lim Zq=nu..} reduces to
\begin{align}\label{eqn:qZ as t-core}
	\lim_{q\rightarrow1^-}
	\frac{Z(Q;Q_1,q)}{\prod_{j,k=1}^\infty (1-Q_1 q^{-j-k+1})}|_{q\rightarrow\xi_tq, Q_1\rightarrow q^t}
	=\sum_{\nu\in\mathcal{P}_{\text{t-core}}} Q^{|\nu|},
\end{align}
which is precisely equation \eqref{eqn:t-core=Zq lim}.

A similar method can be directly applied to prove the case of $n$-point function.
We omit the details and only list the result of each step.
First,
using the formula \eqref{eqn:def C},
the $q$-deformed $n$-point function can be rewritten as 
\begin{align*}
	\begin{split}
	Z(Q;Q_1,&q;s_1,\dots,s_n)
	=\frac{1}{Z(Q;Q_1,q)}\\
	&\cdot \sum_{\mu,\nu\in\mathcal{P}}
	\bigg((-Q_1)^{|\mu|}
	(-Q)^{|\nu|}
	s_{\nu^t}(q^{\rho})s_{\mu^t}(q^{\nu^t+\rho})
	s_{\nu}(q^{\rho})s_{\mu}(q^{\nu+\rho})
	\cdot \prod_{j=1}^n \sum_{i=1}^\infty s_j^{\nu_i-i+\half}\bigg).
	\end{split}
\end{align*}
Then after applying Lemma \ref{lem:from Y23} to take the summation over $\mu$, and using the evaluation formula \eqref{eqn:schur q^rho},
we arrive at
\begin{align}
	\begin{split}
	Z(Q;Q_1,&q;s_1,\dots,s_n)
	=\frac{\prod_{j,k=1}^\infty (1-Q_1 q^{-j-k+1})}{Z(Q;Q_1,q)}\\
	&\cdot\sum_{\nu\in\mathcal{P}}
	\Bigg((QQ_1)^{|\nu|}
	\prod_{(j,k)\in\nu}\frac{(1-Q_1q^{h(j,k)})(1-Q_1^{-1}q^{h(j,k)})} {(1-q^{h(j,k)})^2}
	\cdot \prod_{j=1}^n \sum_{i=1}^\infty s_j^{\nu_i-i+\half}\bigg).
	\end{split}
\end{align}
Note that the change of variables $q\rightarrow\xi_tq,  Q_1\rightarrow q^t$ and taking the limit $q\rightarrow1^-$ do not affect the $\prod_{j=1}^n \sum_{i=1}^\infty s_j^{\nu_i-i+\half}$ part.
Consequently,
from previous discussions performed in equation \eqref{eqn:dis lim} we obtain
\begin{align*}
	\lim_{q\rightarrow1^-} Z(Q;Q_1,q;s_1,\dots,s_n)|_{q\rightarrow\xi_tq \atop Q_1\rightarrow q^t}
	=\frac{1}{\sum\limits_{\nu\in\mathcal{P}_{\text{t-core}}} Q^{|\nu|}}
	\cdot \sum_{\nu\in\mathcal{P}_{\text{t-core}}} \bigg(Q^{|\nu|}\cdot \prod_{j=1}^n \sum_{i=1}^\infty s_j^{\nu_i-i+\half}\bigg),
\end{align*}
which exactly matches with the equation \eqref{eqn:qZn lim Zn}.
\end{proof}

The following two combinatorial lemmas are needed in the above proof.
\begin{lem}\label{lem:from Y23}
	For two fixed partitions $\nu^1, \nu^2\in\mathcal{P}$, we have the following identity
	\begin{align}\label{eqn:inf/inf=finite}
		\begin{split}
			\sum_{\lambda}
			z^{|\lambda|}
			s_{\lambda}(q^{\nu^1+\rho})
			&s_{\lambda^t}(q^{\nu^{2,t}+\rho})
			=\prod_{j,k=1}^\infty (1+zq^{-j-k+1})\\
			&\cdot\prod_{(j,k)\in\nu^1}(1+zq^{\nu_j^1+\nu^{2,t}_k-j-k+1})
			\cdot\prod_{(j,k)\in\nu^2}(1+zq^{-\nu^{1,t}_k-\nu^{2,t}_j+j+k-1}),
		\end{split}
	\end{align}
	where $\nu^{i,t}$ means the transpose of $\nu^i$ for $i=1,2$.
\end{lem}
\begin{proof}
	This lemma is proved by the author in Lemma 3.7 of \cite{Y25}.
\end{proof}

\begin{lem}\label{lem:comb lem for nu}
	For any fixed partition $\nu\in\mathcal{P}$,
	we have
	\begin{align}\label{eqn:sum h=k||}
		\sum_{(j,k)\in\nu} h(j,k)
		=n(\nu)+n(\nu^t)+|\nu|,
	\end{align}
	where for convenience we recall that $n(\nu)=\sum_{i=1}^{l(\nu)} (i-1)\nu_i$.
\end{lem}
\begin{proof}
	The proof is straightforward.
	We compute both sides of equation \eqref{eqn:sum h=k||} and show that they are equal.
	On the one hand,
	from the definition we have
	\begin{align}\label{eqn:n+n+||}
		\begin{split}
		n(\nu)+n(\nu^t)+|\nu|
		=&\sum_{i=1}^{l(\nu)} (i-1)\nu_i
		+\sum_{i=1}^{l(\nu^t)} (i-1)\nu^t_i+|\nu|\\
		=&\sum_{i=1}^{l(\nu)} (i-1)\nu_i 
		+\sum_{i=1}^{l(\nu)}\sum_{j=1}^{\nu_i} j
		=\sum_{i=1}^{l(\nu)} \frac{\nu_i^2+2i\nu_i-\nu_i}{2}.
		\end{split}
	\end{align}
	
	On the other hand,
	\begin{align}\label{eqn:last sum}
		\begin{split}
		\sum_{u\in\nu} h(u)
		=\sum_{i=1}^{l(\nu)}\sum_{j=1}^{\nu_i}
		(\nu_i+\nu_j^t-i-j+1)
		=\sum_{i=1}^{l(\nu)}
		\frac{\nu_i^2-2i\nu_i+\nu_i}{2}
		+\sum_{i=1}^{l(\nu)}\sum_{j=1}^{\nu_i}\nu^t_j.
		\end{split}
	\end{align}
	For the last summand in the above equation,
	we consider the following identity
	\begin{align*}
		\sum_{i=1}^{l(\nu)}\sum_{j=1}^{\nu_i}(2i-1)
		=\sum_{j=1}^{l(\nu^t)} \sum_{i=1}^{\nu^t_j}(2i-1)
		=\sum_{j=1}^{l(\nu^t)} (\nu^t_j)^2
		=\sum_{j=1}^{l(\nu^t)}\sum_{i=1}^{\nu^t_j}\nu^t_j,
	\end{align*}
	Note that this is exactly equal to the last term in equation \eqref{eqn:last sum}.
	Thus, we have
	\begin{align*}
		\sum_{u\in\nu} h(u)
		=\sum_{i=1}^{l(\nu)}
		\frac{\nu_i^2-2i\nu_i+\nu_i}{2}
		+\sum_{i=1}^{l(\nu)}\sum_{j=1}^{\nu_i}(2i-1)
		=\sum_{i=1}^{l(\nu)}
		\frac{\nu_i^2+2i\nu_i-\nu_i}{2}.
	\end{align*}
	Comparing with equation \eqref{eqn:n+n+||},
	we obtain the desired identity.
\end{proof}

As we have promised,
the $q$-deformed $n$-point function $Z(Q;Q_1,q;s_1,\dots,s_n)$ is also a generalization of the ordinary $n$-point function of all integer partitions.
We state this result as follows.
\begin{prop}\label{prop:qZn lim F}
	Denote by
	\begin{align}\label{eqn:qZn lim F}
		F(Q;s_1,\dots,s_n)
		=\frac{1}{\sum\limits_{\nu\in\mathcal{P}} Q^{|\nu|}}
		\cdot \sum_{\nu\in\mathcal{P}} \bigg(\prod_{j=1}^n \sum_{i=1}^\infty s_j^{\nu_i-i+\half}
		\cdot Q^{|\nu|}\bigg)
	\end{align}
	the ordinary $n$-point function of all integer partitions studied by Bloch and Okounkov in \cite{BO}.
	Then we have
	\begin{align*}
		F(Q;s_1,\dots,s_n)
		=\lim_{q\rightarrow \infty} Z(Q;Q_1,q;s_1,\dots,s_n)|_{Q_1=1}.
	\end{align*}
\end{prop}
\begin{proof}
In the proof of Theorem \ref{prop:qZn lim Zn},
we used properties of topological vertex to derive the following two formulas for the $q$-deformed partition function $Z(Q;Q_1,q)$ 
and the $q$-deformed $n$-point function $Z(Q;Q_1,q;s_1,\dots,s_n)$,
\begin{align}\label{eqn:Zq as 1-second}
	\begin{split}
		\frac{Z(Q;Q_1,q)}{\prod_{j,k=1}^\infty (1-Q_1 q^{-j-k+1})}
		=&\sum_{\nu\in\mathcal{P}}
		\bigg((QQ_1)^{|\nu|}
		\prod_{(j,k)\in\nu}\frac{(1-Q_1q^{h(j,k)})(1-Q_1^{-1}q^{h(j,k)})} {(1-q^{h(j,k)})^2}\bigg)
	\end{split}
\end{align}
and
\begin{align}\label{eqn:qZn=-second}
	\begin{split}
		Z(Q;Q_1,&q;s_1,\dots,s_n)
		=\frac{\prod_{j,k=1}^\infty (1-Q_1 q^{-j-k+1})}{Z(Q;Q_1,q)}\\
		&\cdot\sum_{\nu\in\mathcal{P}}
		\Bigg((QQ_1)^{|\nu|}
		\prod_{(j,k)\in\nu}\frac{(1-Q_1q^{h(j,k)})(1-Q_1^{-1}q^{h(j,k)})} {(1-q^{h(j,k)})^2}
		\cdot \prod_{j=1}^n \sum_{i=1}^\infty s_j^{\nu_i-i+\half}\bigg).
	\end{split}
\end{align}
Then,
by setting $Q_1=1$ and $q\rightarrow\infty$,
the term
$$(QQ_1)^{|\nu|}
\prod_{(j,k)\in\nu}\frac{(1-Q_1q^{h(j,k)})(1-Q_1^{-1}q^{h(j,k)})} {(1-q^{h(j,k)})^2}$$
becomes $Q^{|\nu|}$ for any partition $\nu$.
Thus equations \eqref{eqn:Zq as 1-second}, \eqref{eqn:qZn=-second} imply the formula \eqref{eqn:qZn lim F}.
\end{proof}

\subsection{The infinite product formula for the $q$-deformed $n$-point function}
In this subsection,
we prove a formula for the $q$-deformed $n$-point function $Z(Q;Q_1,q;s_1,\dots,s_n)$ expressed as a certain coefficient of an infinite product.
As an application,
we also give a new proof of the formula \eqref{eqn:size-generating function t-core lambda} about the generating function for the number of $t$-core partitions.

We need to use the operator $T(s)$ on the fermionic Fock space $\mathcal{F}$ (see \cite{O01}),
which is defined by
\begin{align*}
	T(s)=\sum_{k\in\mathbb{Z}+\frac{1}{2}}
	s^k \psi_{k}\psi^*_k.
\end{align*}
There are at least two reasons to consider the operator $T(s)$.
On the one hand,
from equation \eqref{eqn:psipsi* action},
we have
\begin{align}\label{eqn:T action}
	T(s) \cdot |\nu\rangle
	=\sum_{i=1}^\infty s^{\nu_i-i+\frac{1}{2}}
	\cdot |\nu\rangle
\end{align}
for any partition $\nu\in\mathcal{P}$.
This helps us to transfer the $q$-deformed $n$-point function $Z(Q;Q_1,q;s_1,\dots,s_n)$ to a vacuum expectation value.
On the other hand,
we have the simple relation between this operator $T(s)$ and fermionic fields as
\begin{align}
	T(s)=[w^0]\ \psi(sw)\psi^*(w),
\end{align}
where the notation $[w^0]$ means taking the coefficient of $w^0$ in the corresponding function/operator.
This can help us to compute the corresponding vacuum expectation value in this subsection and the next subsection.

\begin{prop}\label{prop:qZn as []}
	We have
	\begin{align}\label{eqn:qZn as []}
		\begin{split}
		Z(Q;Q_1,q;s_1,\dots,s_n)
		=\prod_{j=1}^n \frac{\sqrt{s_j^{-1}}}{1-s_j^{-1}}
		\cdot
		[w_1^0\dots w_n^0]
		\ E(Q;Q_1,q;\mathbf{sw},\mathbf{w}),
		\end{split}
	\end{align}
	where we abbreviate the notations $\mathbf{sw}=(s_1w_1,\dots,s_nw_n)$ and
	$\mathbf{w}=(w_1,\dots,w_n)$.
	The function $E(Q;Q_1,q;\mathbf{z},\mathbf{w})$ is given by
	\begin{align}\label{eqn:def E first}
		\begin{split}
		&\prod_{1\leq i<k\leq n} \frac{(1-z_i^{-1}z_k)(1-w_i^{-1}w_k)}{(1-z_i^{-1}w_k)(1-w_i^{-1}z_k)}
		\cdot \prod_{i,k=1}^n \frac{\big(Q_1Qz_i^{-1}z_k,Q_1Qw_i^{-1}w_k;Q_1Q\big)_{\infty}}
		{\big(Q_1Qz_i^{-1}w_k,Q_1Qw_i^{-1}z_k;Q_1Q\big)_{\infty}}\\
		&\qquad\cdot \prod_{j=1}^n\prod_{a=1}^\infty \frac{\big(-Q_1Qz_jq^{-a+\half},-Qw_jq^{-a+\half},-Q_1z_j^{-1}q^{-a+\half},-w_j^{-1}q^{-a+\half};Q_1Q\big)_{\infty}}
		{\big(-Qz_jq^{-a+\half},-Q_1Qw_jq^{-a+\half},-z_j^{-1}q^{-a+\half},-Q_1w_j^{-1}q^{-a+\half};Q_1Q\big)_{\infty}},
		\end{split}
	\end{align}
	where
	$(z;Q)_{\infty} :=\prod_{b=0}^\infty (1-z Q^b)$ and
	$(z_1,\cdots,z_k;Q)_{\infty}:=\prod_{i=1}^k (z_i;Q)_{\infty}$,
	and we restrict ourselves to the region
	\begin{align}\label{eqn:region}
		|Q_1|, |Q|\ll 1,
		\qquad 1<|w_n|<|s_nw_n|<\dots<|w_1|<|s_1w_1|< |Q^{-1}|.
	\end{align}
\end{prop}
\begin{proof}
Starting from the definition \eqref{eqn:def qZn} of the $q$-deformed $n$-point function and using the rotation symmetry  \eqref{eqn:rotation symm} of the topological vertex,
we rewrite it as
\begin{align}\label{eqn:top vet as vev}
	\begin{split}
	Z(Q;Q_1,q;s_1,&\dots,s_n)
	=\frac{1}{Z(Q;Q_1,q)}\\
	&\cdot\sum_{\mu,\nu\in\mathcal{P}}
	\bigg((-Q_1)^{|\mu|}
	(-Q)^{|\nu|}
	C_{\mu^t,\nu,\emptyset}(q)
	C_{\mu,\nu^t,\emptyset}(q)
	\cdot \prod_{j=1}^n \sum_{i=1}^\infty s_j^{\nu_i-i+\half}\bigg).
	\end{split}
\end{align}
Then we apply the vacuum expectation value formula \eqref{eqn:clm0=<>} for the topological vertex.
Consequently,
the second line of the above equation \eqref{eqn:top vet as vev} becomes
\begin{align}\label{eqn:qZn as <><>}
	&\sum_{\mu,\nu\in\mathcal{P}}
	\Big((-Q_1)^{|\mu|}
	(-Q)^{|\nu|}\langle\mu|\Gamma_-(q^{\rho})\Gamma_+(q^{\rho})
	|\nu\rangle
	\langle\mu^t|\Gamma_-(q^{\rho})\Gamma_+(q^{\rho})
	|\nu^t\rangle
	\cdot \prod_{j=1}^n \sum_{i=1}^\infty s_j^{\nu_i-i+\half}\Big).
\end{align}
To simplify the above formula,
we use the following identity
\begin{align*}
	\langle\mu^t|\Gamma_-(q^{\rho})\Gamma_+(q^{\rho})
	|\nu^t\rangle
	=\langle\nu^t|\Gamma_-(q^{\rho})\Gamma_+(q^{\rho})
	|\mu^t\rangle
	=\langle\nu|\Gamma_-(-q^{\rho})^{-1}\Gamma_+(-q^{\rho})^{-1}
	|\mu\rangle,
\end{align*}
where the first equal sign follows from taking dual (see equation \eqref{eqn:taking dual}),
and the second equal sign is obtained from taking the $\omega$-transform (see equation \eqref{eqn:taking w}).
We also need the following two equations
\begin{align*}
	Q^{L_0}\cdot|\nu\rangle
	=Q^{|\nu|}\cdot|\nu\rangle
	\quad\text{and} \quad
	\prod_{j=1}^n T(s_j) \cdot|\nu\rangle
	=\prod_{j=1}^n \sum_{i=1}^\infty s_j^{\nu_i-i+\half} \cdot |\nu\rangle,
\end{align*}
where the first follows from equation \eqref{eqn:action L0} and the second from \eqref{eqn:T action}.
Then equation \eqref{eqn:qZn as <><>} becomes
\begin{align*}
	&\sum_{\mu,\nu\in\mathcal{P}}
	\Big(\langle\mu|\Gamma_-(q^{\rho})\Gamma_+(q^{\rho})
	(-Q)^{L_0}\prod_{j=1}^n T(s_j)|\nu\rangle
	\cdot \langle\nu|\Gamma_-(-q^{\rho})^{-1}\Gamma_+(-q^{\rho})^{-1}
	(-Q_1)^{L_0}|\mu\rangle\Big).
\end{align*}
Since $\{|\mu\rangle\}_{\mu\in\mathcal{P}}$ is a basis of the charge zero fermionic Fock space $\mathcal{F}^{(0)}$, 
the operator $\sum_{\mu\in\mathcal{P}} |\mu\rangle\langle\mu|$ could be regarded as an identity operator on this space.
Finally, we have
\begin{align}\label{eqn:qZn as <>}
	\begin{split}
	Z(Q;Q_1,q;s_1,\dots,s_n)=
	\frac{1}{Z(Q;Q_1,q)}
	\cdot\sum_{\nu\in\mathcal{P}}
	\langle\nu| G_0
	\prod_{j=1}^n T(s_j) |\nu\rangle,
	\end{split}
\end{align}
where we denote
\begin{align}\label{def:G_0}
	G_0=\Gamma_-(-q^{\rho})^{-1}\Gamma_+(-q^{\rho})^{-1}
	(-Q_1)^{L_0}\Gamma_-(q^{\rho})\Gamma_+(q^{\rho})
	(-Q)^{L_0}.
\end{align}
This gives a vacuum expectation value formula for the $q$-deformed $n$-point function.

Motivated by the above vacuum expectation value formula \eqref{eqn:qZn as <>} and the relation between the operator $T(s)$ and the fermionic fields $\psi(z), \psi^*(w)$,
we introduce the following new function
\begin{align}\label{eqn:def E}
	\begin{split}
	\bar{E}(Q;Q_1,q;\mathbf{z},\mathbf{w})
	:=\sum_{\nu\in\mathcal{P}}
	\langle\nu|G_0
	\prod_{j=1}^n\Big(\psi(z_i)\psi^*(w_i)\Big)|\nu\rangle.
	\end{split}
\end{align}
For simplicity,
we use the notation $\bar{E}(Q;Q_1,q)$ to denote the $n=0$ case of $\bar{E}(Q;Q_1,q;\mathbf{z},\mathbf{w})$.
Indeed, we have $\bar{E}(Q;Q_1,q)=Z(Q;Q_1,q)$.
Then from the relation $T(s)=[w^0]\psi(sw)\psi^*(w)$
and equation \eqref{eqn:qZn as <>},
we have
\begin{align}\label{eqn:qZn as []E}
	Z(Q;Q_1,q;s_1,\dots,s_n)
	=[w_1^0\dots w_n^0]
	\ \frac{\bar{E}(Q;Q_1,q;\mathbf{sw},\mathbf{w})}
	{\bar{E}(Q;Q_1,q)}.
\end{align}
Thus,
to obtain an explicit formula for the $q$-deformed $n$-point function $Z(Q;Q_1,q;\mathbf{s})$,
we need to compute the function $\bar{E}(Q;Q_1,q;\mathbf{z},\mathbf{w})$ defined in equation \eqref{eqn:def E}.
Such vacuum expectation values are widely studied in the literature.
See, for examples, \cite{Eng21,EO06,Mil03,O01,Y25}.
Below,
we derive an infinite product formula for the function $\bar{E}(Q;Q_1,q;\mathbf{z},\mathbf{w})$,
which proves this proposition.
To do that,
we use the formula \eqref{eqn:psi as gamma} to represent the fermionic fields in equation \eqref{eqn:def E}  in terms of vertex operators.
Then we have
\begin{align*}
	\bar{E}(Q;Q_1,q;\mathbf{z},\mathbf{w})
	=\prod_{j=1}^n \sqrt{w_j/z_j}
	\cdot \sum_{\nu\in\mathcal{P}}
	\langle\nu|G_0
	\prod_{j=1}^n \Big(\Gamma_-(z_j)\Gamma_+(z_j^{-1})^{-1}
	\Gamma_-(w_j)^{-1} \Gamma_+(w_j^{-1}) \Big)|\nu\rangle.
\end{align*}
Using the commutation relations \eqref{eqn:comm Gammapminf} to reorder the operators in the above vacuum expectation value,
and restricting ourselves to the following region
\begin{align}
	|Q_1|,|Q|\ll 1,
	\qquad 1<|w_n|<|z_n|<\dots<|w_1|<|z_{1}|< |Q^{-1}|,
\end{align}
we move all the operators $\Gamma_-(\cdot)$ to the left,
and all $\Gamma_+(\cdot)$ to the right.
As a result,
\begin{align}\label{eqn:E as E_{VEV}}
	\begin{split}
	\bar{E}(Q;Q_1,q;\mathbf{z},\mathbf{w})
	=&\prod_{a=1}^\infty (1-Q_1q^{-a})^a
	\cdot \prod_{j=1}^n\prod_{a=1}^\infty \frac{(1+Qw_jq^{-a+\half})(1+Q_1Qz_jq^{-a+\half})}{(1+Qz_jq^{-a+\half})(1+Q_1Qw_jq^{-a+\half})}\\
	&\cdot \prod_{j=1}^n \frac{\sqrt{w_j/z_j}}{1-z_j^{-1}w_j}
	\cdot\prod_{1\leq i<k\leq n} \frac{(1-z_i^{-1}z_k)(1-w_i^{-1}w_k)}{(1-z_i^{-1}w_k)(1-w_i^{-1}z_k)}
	\cdot E_{VEV},
	\end{split}
\end{align}
where
\begin{align*}
	E_{VEV}:=\sum_{\nu\in\mathcal{P}}
	\langle\nu|&\Gamma_-(-q^{\rho})^{-1}\Gamma_-(-Q_1q^{\rho})
	\cdot \prod_{j=1}^n \Big(\Gamma_-(Q_1Qz_j)
	\Gamma_-(Q_1Qw_j)^{-1}\Big)\\
	&\cdot (Q_1Q)^{L_0}
	\cdot \Gamma_+(-Q_1Qq^{\rho})^{-1}\Gamma_+(-Qq^{\rho})
	\cdot \prod_{j=1}^n \Big( \Gamma_+(z_j^{-1})^{-1}\Gamma_+(w_j^{-1}) \Big)|\nu\rangle
\end{align*}
Note that
since $\{|\nu\rangle\}_{\nu\in\mathcal{P}}$ is a basis of the fermionic Fock space $\mathcal{F}^{(0)}$,
and $\{\langle\nu|\}_{\nu\in\mathcal{P}}$ is its dual basis,
the above equation is a trace over the space $\mathcal{F}^{(0)}$.
Thus,
one can first move all the operators $\Gamma_+(\cdot)$ to the far left side using the cyclic property of the trace,
and then use the commutation relations \eqref{eqn:comm Gammapminf} again to move them to the right.
The same is done for the operators $\Gamma_-(\cdot)$.
Iterating this process yields the following infinite product formula for the function $E_{VEV}$
\begin{align}\label{eqn: for Evev}
	\begin{split}
		E_{VEV}
		=&\prod_{b=1}^\infty
		\frac{M\big((Q_1Q)^{b}Q_1;q\big)M\big((Q_1Q)^{b-1}Q;q\big)}{(1-(Q_1Q)^{b})\cdot M\big((Q_1Q)^{b};q\big)^2}
		\cdot \prod_{i,k=1}^n \frac{\big(Q_1Qz_i^{-1}z_k,Q_1Qw_i^{-1}w_k;Q_1Q\big)_{\infty}}
		{\big(Q_1Qz_i^{-1}w_k,Q_1Qw_i^{-1}z_k;Q_1Q\big)_{\infty}}\\
		&\cdot \prod_{j=1}^n\prod_{a=1}^\infty \frac{\big(-Q_1^2Q^2z_jq^{-a+\half},-Q_1Q^2w_jq^{-a+\half},-Q_1z_j^{-1}q^{-a+\half},-w_j^{-1}q^{-a+\half};Q_1Q\big)_{\infty}}
		{\big(-Q_1Q^2z_jq^{-a+\half},-Q_1^2Q^2w_jq^{-a+\half},-z_j^{-1}q^{-a+\half},-Q_1w_j^{-1}q^{-a+\half};Q_1Q\big)_{\infty}},
	\end{split}
\end{align}
where $M(\cdot;\cdot)$ is the MacMahon formula given by $M(z;q)=\prod_{j=1}^\infty(1-zq^{-j})^j$,
and we have used the formula
\begin{align*}
	\sum_{\nu\in\mathcal{P}}
	\langle\nu| (Q_1Q)^{L_0} |\nu\rangle
	=\sum_{\nu\in\mathcal{P}}
	(Q_1Q)^{|\nu|}
	=\prod_{b=1}^{\infty} \frac{1}{1-(Q_1Q)^b},
\end{align*}
and the fact that
as operators on the charge zero fermionic Fock space $\mathcal{F}^{(0)}$,
$$\Gamma_+(Q^M\cdot)^{\pm},\ \Gamma_-(Q^M\cdot)^{\pm} \rightarrow \text{id},
\text{\ when\ } M\rightarrow \infty.$$

Combing the above formula \eqref{eqn: for Evev} of $E_{VEV}$ and equation \eqref{eqn:E as E_{VEV}},
we can obtain an infinite product formula for the function $\bar{E}(Q;Q_1,q;\mathbf{z},\mathbf{w})$.
We omit the details and
in particular, for the case of $n=0$, we have
\begin{align}\label{eqn:qZ formula n=0}
	Z(Q;Q_1,q)
	=\bar{E}(Q;Q_1,q)
	=M(Q_1;q)
	\cdot\prod_{b=1}^\infty
	\frac{M\big((Q_1Q)^{b}Q_1;q\big)\cdot M\big((Q_1Q)^{b-1}Q;q\big)}
	{(1-(Q_1Q)^b)\cdot M\big((Q_1Q)^{b};q\big)^2}.
\end{align}
By comparing equations \eqref{eqn:E as E_{VEV}}, \eqref{eqn: for Evev} with the function $E(Q;Q_1,q;\mathbf{z},\mathbf{w})$ defined in equation \eqref{eqn:def E first},
we have
\begin{align}\label{eqn:barE and E}
	\frac{\bar{E}(Q;Q_1,q;\mathbf{z},\mathbf{w})}
	{\bar{E}(Q;Q_1,q)}
	=E(Q;Q_1,q;\mathbf{z},\mathbf{w})
	\cdot\prod_{j=1}^n \frac{\sqrt{w_j/z_j}}{1-z_j^{-1}w_j}.
\end{align}
Together with formula \eqref{eqn:qZn as []E},
the equation \eqref{eqn:qZn as []} is proved.
\end{proof}

\begin{rmk}\label{rmk:new pf Z}
	Note that in equation \eqref{eqn:qZ as t-core},
	we showed that the generating function for the number of $t$-core partitions can be expressed in terms of the $q$-deformed partition function,
	which is computed in equation \eqref{eqn:qZ formula n=0}.
	Thus we have
	\begin{align}
		\begin{split}
		\sum_{\nu\in\mathcal{P}_{\text{t-core}}} Q^{|\nu|}
		=&\lim_{q\rightarrow1^-}
		\frac{Z(Q;Q_1,q)}{\prod_{j,k=1}^\infty (1-Q_1 q^{-j-k+1})}|_{q\rightarrow\xi_tq, Q_1\rightarrow q^t}\\
		=&\lim_{q\rightarrow1^-}
		\prod_{b=1}^\infty
		\frac{\prod_{n=1}^\infty (1-Q^bq^{(b+1)t-n})^n(1-Q^{b}q^{(b-1)t-n})^n}
		{(1-Q^bq^{bt}) \prod_{n=1}^\infty (1-Q^bq^{bt-n})^{2n}}|_{q\rightarrow\xi_tq}\\
		=&\prod_{b=1}^\infty
		\frac{\prod_{n=1}^t(1-Q^b\xi_t^{-n})^t}
		{(1-Q^b)}
		=\prod_{b=1}^\infty
		\frac{(1-Q^{bt})^t}
		{(1-Q^b)},
		\end{split}
	\end{align}
	which gives a new proof of the formula \eqref{eqn:size-generating function t-core lambda}.
\end{rmk}

\subsection{The determinant formula for the $q$-deformed $n$-point function}
In this subsection,
using a version of Wick Theorem proved in \cite{WYZ25},
we give another formula for the $q$-deformed $n$-point function $Z(Q;Q_1,q;s_1,\dots,s_n)$ as a certain coefficient of a matrix determinant.
This formula will be used to derive the closed formula for the $n$-point function of $t$-core partitions.

We first recall the following version of Wick Theorem proved by Wang, Zhou and the author in Lemma 5.1 of \cite{WYZ25}, which is a generalization of Lemma B.1 in \cite{BB19},
\begin{lem}\label{lem:Wick}
Let $\varphi_1, \cdots, \varphi_n$
be linear combinations of the free fermions $\{\psi_r\}_{r\in \bZ +\half}$,
and let $\varphi_1^*, \cdots, \varphi_n^* $
be linear combinations of $\{\psi_r^*\}_{r\in \bZ +\half}$.
Suppose that $G_1,\cdots,G_n,G_{n+1}$ are some fermionic operators satisfying the
Hirota bilinear relation
\begin{align}\label{eqn:Hirota}
	[G\otimes G, \sum_{r\in \bZ+\half} \psi_r^* \otimes \psi_{r}] = 0,
\end{align}
and assume that
$\sum_{N\in \bZ}\sum_{\mu\in \cP} \langle\mu| R^N \cdot
G_1  G_2 \cdots  G_{n+1}
\cdot R^{-N} |\mu\rangle \not= 0$.
Then we have
\begin{align}
\label{eq-lem-det}
\begin{split}
	 \frac{\sum_{N\in \bZ}\sum_{\mu\in \cP} \langle\mu| R^N \cdot
		G_1 \varphi_1\varphi_1^* G_2 \varphi_2\varphi_2^* \cdots  \varphi_n\varphi_n^* G_{n+1}
		\cdot R^{-N} |\mu\rangle}
	{\sum_{N\in \bZ}\sum_{\mu\in \cP} \langle\mu| R^N \cdot
		G_1  G_2 \cdots  G_{n+1}
		\cdot R^{-N} |\mu\rangle} 
	= \det(A_{ij})_{i,j=1}^n,
\end{split}
\end{align}
where
\begin{equation*}
	\begin{split}
		A_{ij} =\begin{cases}
			\frac{\sum_{N\in \bZ}\sum_{\mu\in \cP} \langle\mu| R^N \cdot
				G_1 G_2 \cdots G_i\varphi_iG_{i+1}
				\cdots G_{j} \varphi_j^* G_{j+1}  \cdots G_{n+1}
				\cdot R^{-N} |\mu\rangle}
			{\sum_{N\in \bZ}\sum_{\mu\in \cP} \langle\mu| R^N \cdot
				G_1  G_2 \cdots  G_{n+1}
				\cdot R^{-N} |\mu\rangle}, & \text{if\ }i\leq j,\\
				-\frac{\sum_{N\in \bZ}\sum_{\mu\in \cP} \langle\mu| R^N \cdot
				G_1 G_2 \cdots G_j\varphi_jG_{j+1}
				\cdots G_{i} \varphi_i^* G_{i+1}  \cdots G_{n+1}
				\cdot R^{-N} |\mu\rangle}
				{\sum_{N\in \bZ}\sum_{\mu\in \cP} \langle\mu| R^N \cdot
				G_1  G_2 \cdots  G_{n+1}
				\cdot R^{-N} |\mu\rangle}, & \text{if\ }i>j.
		\end{cases}
	\end{split}
\end{equation*}
\end{lem}

The main result of this subsection is the following formula for the $q$-deformed $n$-point function $Z(Q;Q_1,q;s_1,\dots,s_n)$.
\begin{prop}\label{thm:qZn as det}
	We have
	\begin{align}\label{eqn:qZn as det}
		\begin{split}
		Z(Q;Q_1,q;s_1,&\dots,s_n)
		=\frac{s_{[n]}^{-1/2}}{\Theta_3(Q_2;Q_1Q)^{n-1}\cdot \Theta_3(Q_2s_{[n]}^{-1};Q_1Q)}\\
		&\cdot[w_1^0\dots w_n^0]\det \Big(\frac{\Theta_3\big(Q_2s_i^{-1}w_i^{-1}w_j;Q_1Q\big)E(Q;Q_1,q;s_iw_i,w_j)}
		{(1-s_i^{-1}w_i^{-1}w_j)}\Big)_{i,j=1}^n,
		\end{split}
	\end{align}
	where we have used $s_{[n]}:=\prod_{j=1}^n s_j$, $\Theta_3(z;Q):=\sum_{a\in\mathbb{Z}} z^a Q^{\frac{a^2}{2}}$,
	and we recall that the function $E(Q;Q_1,q;z,w)$ satisfies ($n=1$ case of equation \eqref{eqn:def E first})
	\begin{align}\label{eqn:def E first 1-variable}
		\begin{split}
			E(Q;&Q_1,q;z,w)= \frac{\big(Q_1Q;Q_1Q\big)_{\infty}^2}
			{\big(Q_1Qz^{-1}w,Q_1Qw^{-1}z;Q_1Q\big)_{\infty}}\\
			&\cdot \prod_{a=1}^\infty \frac{\big(-Q_1Qzq^{-a+\half},-Qwq^{-a+\half},-Q_1z^{-1}q^{-a+\half},-w^{-1}q^{-a+\half};Q_1Q\big)_{\infty}}
			{\big(-Qzq^{-a+\half},-Q_1Qwq^{-a+\half},-z^{-1}q^{-a+\half},-Q_1w^{-1}q^{-a+\half};Q_1Q\big)_{\infty}}.
		\end{split}
	\end{align}
\end{prop}
\begin{proof}
In the proof of Proposition \ref{prop:qZn as []},
we have obtained (see equation \eqref{eqn:qZn as []E})
\begin{align}\label{eqn:qZn as []Edet}
	Z(Q;Q_1,q;s_1,\dots,s_n)
	=[w_1^0\dots w_n^0]
	\frac{\bar{E}(Q;Q_1,q;\mathbf{sw},\mathbf{w})}
	{\bar{E}(Q;Q_1,q)},
\end{align}
where 

\begin{align}\label{eqn:def Edet}
	\begin{split}
		\bar{E}(Q;Q_1,q;\mathbf{z},\mathbf{w})
		=\sum_{\nu\in\mathcal{P}}
		\langle\nu|G_0
		\prod_{j=1}^n\Big(\psi(z_i)\psi^*(w_i)\Big)|\nu\rangle,
	\end{split}
\end{align}
and the operator $G_0$ is given in equation \eqref{def:G_0}.
Notice that,
the function
$\bar{E}(Q;Q_1,q;\mathbf{z},\mathbf{w})$ is a trace of certain vertex operators and fermionic fields over the charge zero fermionic Fock space $\mathcal{F}^{(0)}$.
To apply Lemma \ref{lem:Wick},
we introduce the following new function
\begin{align}\label{eqn:def tildeEdet}
	\begin{split}
		\tilde{E}(Q;Q_1,Q_2,q;\mathbf{z},\mathbf{w})
		:=
		\sum_{N\in\mathbb{Z}}
		\sum_{\nu\in\mathcal{P}}
		\langle\nu|R^NQ_2^{-C}G_0
		\prod_{j=1}^n\big(\psi(z_i)\psi^*(w_i)\big)
		R^{-N}|\nu\rangle,
	\end{split}
\end{align}
which is a trace over the full fermionic Fock space $\mathcal{F}$.
Then, using the commutation relations \eqref{eqn:comm RCH} and \eqref{eqn:comm RGamma} for the operators $R^N$ and $C, H, \Gamma_\pm(z), \psi(z), \psi^*(w)$,
we obtain the following relation between functions $\tilde{E}(Q;Q_1,Q_2,q;\mathbf{z},\mathbf{w})$
and $\bar{E}(Q;Q_1,q;\mathbf{z},\mathbf{w})$
\begin{align*}
	\tilde{E}(Q;Q_1,Q_2,q;\mathbf{z},\mathbf{w})
	=&\sum_{N\in\mathbb{Z}}
	(Q_1Q)^{N^2/2} \cdot \Big(Q_2\prod_{j=1}^n (w_j/z_j)\Big)^N
	\cdot \bar{E}(Q;Q_1,q;\mathbf{z},\mathbf{w})\\
	=&\Theta_3\Big(Q_2\prod_{j=1}^n (w_j/z_j);Q_1Q\Big)
	\cdot \bar{E}(Q;Q_1,q;\mathbf{z},\mathbf{w}).
\end{align*}
Thus,
from equation \eqref{eqn:qZn as []Edet},
the $q$-deformed $n$-point function can be represented in terms of $\tilde{E}(Q;Q_1,Q_2,q;\mathbf{z},\mathbf{w})$ by
\begin{align}\label{eqn:qZn as []tildeEdet}
	\begin{split}
	Z(Q;Q_1,q;s_1,\dots,s_n)
	=
	\frac{\Theta_3(Q_2;Q_1Q)}{\Theta_3(Q_2s_{[n]}^{-1};Q_1Q)}
	\cdot[w_1^0\dots w_n^0]
	\frac{\tilde{E}(Q;Q_1,Q_2,q;\mathbf{sw},\mathbf{w})}
	{\tilde{E}(Q;Q_1,Q_2,q)},
	\end{split}
\end{align}
To compute the function $\tilde{E}(Q;Q_1,Q_2,q;\mathbf{z},\mathbf{w})$ appearing in the last term of the above equation,
note that $G_0\in\widehat{GL(\infty)}$ satisfies the Hirota bilinear relation \eqref{eqn:Hirota},
so Lemma \ref{lem:Wick} applies.
Thus
the last term $\tilde{E}(Q;Q_1,Q_2,q;\mathbf{sw},\mathbf{w})/\tilde{E}(Q;Q_1,Q_2,q)$ is equal to the determinant $\det(A_{i,j})_{i,j=1}^n$,
where elements of the matrix $(A_{i,j})_{1\leq i,j \leq n}$ are given by
\begin{align*}
	A_{i,j}=
	\frac{\sum_{N\in\mathbb{Z}}
	\sum_{\nu\in\mathcal{P}}
	\langle\nu|R^NQ_2^{-C}G_0
	\psi(s_iw_i)\psi^*(w_j)
	R^{-N}|\nu\rangle}
	{\sum_{N\in\mathbb{Z}}
		\sum_{\nu\in\mathcal{P}}
		\langle\nu|R^NQ_2^{-C}G_0
		R^{-N}|\nu\rangle}
\end{align*}
if $i\leq j$,
and
\begin{align*}
	A_{i,j}=
	-\frac{\sum_{N\in\mathbb{Z}}
		\sum_{\nu\in\mathcal{P}}
		\langle\nu|R^NQ_2^{-C}G_0
		\psi^*(w_j)\psi(s_iw_i)
		R^{-N}|\nu\rangle}
	{\sum_{N\in\mathbb{Z}}
		\sum_{\nu\in\mathcal{P}}
		\langle\nu|R^NQ_2^{-C}G_0
		R^{-N}|\nu\rangle}
\end{align*}
if $i>j$.
The entries $A_{i,j}, 1\leq i,j \leq n$ can be computed explicitly using the same method used as in Proposition \ref{prop:qZn as []}.
The result is,
for both of the two cases $i\leq j$ and $i>j$,
\begin{align*}
	A_{i,j}
	=\frac{\Theta_3\big(Q_2s_i^{-1}w_i^{-1}w_j;Q_1Q\big)E(Q;Q_1,q;s_iw_i,w_j)\sqrt{s_i^{-1}w_i^{-1}w_j}}
	{\Theta_3(Q_2;Q_1Q)(1-s_i^{-1}w_i^{-1}w_j)},
\end{align*}
where the function $E(Q;Q_1,q;,w_j)$ is given in equation \eqref{eqn:def E first 1-variable}.
Consequently,
following equation \eqref{eqn:qZn as []tildeEdet},
the $q$-deformed $n$-point function has the following formula
\begin{align}\label{eqn:qZn as []detdet}
	\begin{split}
	Z(Q;&Q_1,q;s_1,\dots,s_n)=\frac{\Theta_3(Q_2;Q_1Q)}{\Theta_3(Q_2s_{[n]}^{-1};Q_1Q)}
	\cdot[w_1^0\dots w_n^0]
	\det(A_{i,j})_{i,j=1}^n\\
	=&\frac{s_{[n]}^{-1/2}\Theta_3(Q_2;Q_1Q)}{\Theta_3(Q_2s_{[n]}^{-1};Q_1Q)}
	\cdot[w_1^0\dots w_n^0]
	\det\Bigg(\frac{\Theta_3\big(Q_2s_i^{-1}w_i^{-1}w_j;Q_1Q\big)E(Q;Q_1,q;s_iw_i,w_j)}
	{\Theta_3(Q_2;Q_1Q)(1-s_i^{-1}w_i^{-1}w_j)}\Bigg),
	\end{split}
\end{align}
which is equivalent to the desired equation \eqref{eqn:qZn as det}.
\end{proof}

\section{Closed formula for the $n$-point function of $t$-core partitions}
\label{sec:for Ft}
In this section,
we provide formulas for the $n$-point function of $t$-core partitions in terms of theta functions.

\subsection{Two formulas for the $n$-point function of $t$-core partitions}
As corollaries of the results in the previous section,
we give two formulas for the $n$-point function of $t$-core partitions,
each expressed as a coefficient of a function.

Following the notations in \cite{Eng21,EO06},
the theta function $\vartheta(z;Q)$ is defined to be
\begin{align}\label{eqn:def vartheta}
	\vartheta(z;Q):=&(z^{1/2}-z^{-1/2})
	\cdot \prod_{b=1}^\infty \frac{(1-zQ^b)(1-z^{-1}Q^b)}{(1-Q^b)^2}.
\end{align}
One can consider the following function
\begin{align}\label{eqn:def j}
	j(z;Q)
	:=\sum_{a\in\mathbb{Z}} (-z)^a Q^{\frac{a^2-a}{2}}
	=\prod_{b=1}^\infty (1-Q^b)(1-zQ^{b-1})(1-z^{-1}Q^b),
\end{align}
where the second equal sign follows from the famous Jacobi triple product formula (see \cite{And98}).
For fixed $|Q|<1$,
the above sum over $a\in\mathbb{Z}$ is absolutely convergent.
Thus the function $j(z;Q)$ is an entire function on the non-zero complex plane $z\in\mathbb{C}\setminus\{0\}$.
Moreover,
$z=1$ is a zero of $j(z;Q)$, 
and $j'(1;Q)=\prod_{b=1}^\infty(1-Q^b)^3$.
Then
by comparing the infinite product formulas \eqref{eqn:def vartheta} and \eqref{eqn:def j} of $j(z;Q)$ and $\vartheta(z;Q)$, respectively,
the theta function $\vartheta(z;Q)$ can be represented as
$$\vartheta(z;Q)=z^{-1/2}\cdot \frac{j(z;Q)}{j'(1,Q)}.$$
As a result, $\vartheta(z;Q)$ can be regarded as a two-valued analytic function on $z\in\mathbb{C}\setminus\{0\}$,
and moreover $\vartheta'(1;Q)=1$.
From properties of $j(z;Q)$,
it is easy to verify the following transformation rules satisfied by the theta function $\vartheta(z;Q)$,
\begin{align}\label{eqn:mod theta}
	\vartheta(Qz;Q)=-Q^{-1/2} z^{-1} \vartheta(z;Q),
	\qquad\vartheta(Q^{-1}z;Q)=-Q^{-1/2} z \vartheta(z;Q).
\end{align}
Moreover,
we have $\vartheta'(Q;Q)=-Q^{-3/2}$ and $\vartheta'(Q^{-1};Q)=-Q^{1/2}$.
These values will be frequently used in deriving the closed formula for the $n$-point function of $t$-core partitions in subsection \ref{sec:closed formula}.
We remark that since $\vartheta(z;Q)$ is a two-valued function,
one actually has $\vartheta(e^{2\pi i}z;Q)=-\vartheta(z;Q)$ and $\vartheta(z^{-1};Q)=-\vartheta(z;Q)$,
$\vartheta(-z^{-1};Q)=\vartheta(-z;Q)$.

In the last section,
we already used another theta function as
\begin{align}\label{eqn:def theta3}
	\Theta_3(z;Q):=\sum_{a\in\mathbb{Z}} z^a Q^{\frac{a^2}{2}}
	=\prod_{b=1}^\infty (1-Q^b)(1+zQ^{b-1/2})(1+z^{-1}Q^{b-1/2}).
\end{align}
By comparing equations \eqref{eqn:def j} and \eqref{eqn:def theta3},
we have $\Theta_3(z;Q)=j(-zQ^{1/2};Q)$.
Then still from the properties of $j(z;Q)$,
we have $\Theta_3(Qz;Q)=Q^{-1/2}z^{-1} \Theta_3(z;Q)$.
Notice that the theta function $\Theta_3(z;Q)$ is two-valued with respect to the variable $Q$
and one has $\Theta_3(z;e^{2\pi i}Q)=\Theta_3(-z;Q)$.

Based on Theorem \ref{prop:qZn lim Zn} and Proposition \ref{prop:qZn as []},
the first formula for the $n$-point function of $t$-core partitions can be summarized as the following
\begin{cor}
	For the $n$-point function of $t$-core partitions $F_t(s_1,\dots,s_n)$,
	we have
	\begin{align}\label{eqn:first eqn for F_t}
		\begin{split}
		F_t(Q;s_1,\dots,s_n)
		=&\prod_{j=1}^n \frac{s_j^{-t/2}}
		{\vartheta(s_j)}
		\cdot [w_1^0\dots w_n^0]\\
		&\qquad
		\prod_{i<k}^n \frac{\vartheta(s_i^{-1}s_kw_i^{-1}w_k)
			\vartheta(w_i^{-1}w_k)}
		{\vartheta(s_i^{-1}w_i^{-1}w_k)
			\vartheta(s_kw_i^{-1}w_k)}
		\cdot \prod_{j=1}^n\prod_{a=1}^t
		\frac{\vartheta(-s_jw_j\xi_t^{a})}
		{\vartheta(-w_j\xi_t^{a})}
		\end{split}
	\end{align}
	in the region \eqref{eqn:region},
	where the second argument $Q$ of the theta function will be suppressed when no ambiguity arises.
\end{cor}
\begin{proof}
First,
in equation \eqref{eqn:qZn as []},
since we are only interested in the coefficient of $w_1^0\dots w_n^0$,
we can replace $w_i$ by $q^{\half}w_i$ in the function $E(Q;Q_1,q;\mathbf{sw},\mathbf{w})$.
That is to say,
we have
\begin{align*}
	Z(Q;Q_1,q;s_1,\dots,s_n)
	=\prod_{j=1}^n \frac{\sqrt{s_j^{-1}}}{1-s_j^{-1}}
	\cdot
	[w_1^0\dots w_n^0]
	\ E(Q;Q_1,q;q^{\half}\cdot \mathbf{sw},q^{\half}\cdot \mathbf{w}).
\end{align*}
Then by Theorem \ref{prop:qZn lim Zn},
the $n$-point function of $t$-core partitions $F_t(s_1,\dots,s_n)$ equals
\begin{align*}
	&\lim_{q\rightarrow1^-}
	\Big(\prod_{j=1}^n \frac{\sqrt{s_j^{-1}}}{1-s_j^{-1}}
	\cdot
	[w_1^0\dots w_n^0]
	\ E(Q;Q_1,q;q^{\half}\cdot \mathbf{sw},q^{\half}\cdot \mathbf{w})\Big)|_{q\rightarrow \xi_t q,Q_1\rightarrow q^t},
\end{align*}
which,
together with the explicit formula \eqref{eqn:def E first} for $E(Q;Q_1,q;\mathbf{z},\mathbf{w})$,
can be simplified as
\begin{align*}
	&\prod_{j=1}^n \frac{\sqrt{s_j^{-1}}}{1-s_j^{-1}}
	\cdot [w_1^0\dots w_n^0]\ \lim_{q\rightarrow1^-}
	\bigg(
	 \prod_{1\leq i<k\leq n} \frac{(1-s_i^{-1}s_kw_i^{-1}w_k)(1-w_i^{-1}w_k)}{(1-s_i^{-1}w_i^{-1}w_k)(1-s_kw_i^{-1}w_k)}\\
	&\qquad\cdot\prod_{j=1}^n\prod_{a=1}^\infty \frac{(-Q_1Qs_jw_jq^{-a+1},-Qw_jq^{-a+1},-Q_1s_j^{-1}w_j^{-1}q^{-a},-w_j^{-1}q^{-a};Q_1Q)_{\infty}}
	{(-Qs_jw_jq^{-a+1},-Q_1Qw_jq^{-a+1},-s_j^{-1}w_j^{-1}q^{-a},-Q_1w_j^{-1}q^{-a};Q_1Q)_{\infty}}
	\\
	&\qquad\qquad\qquad\qquad\cdot \prod_{i,k=1}^n \frac{(Q_1Qs_i^{-1}s_kw_i^{-1}w_k,Q_1Qw_i^{-1}w_k;Q_1Q)_{\infty}}
	{(Q_1Qs_i^{-1}w_i^{-1}w_k,Q_1Qs_kw_i^{-1}w_k;Q_1Q)_{\infty}}\bigg)\Big|_{q\rightarrow \xi_t q,Q_1\rightarrow q^t}\\
	=&\prod_{j=1}^n \frac{\sqrt{s_j^{-1}}}{1-s_j^{-1}}
	\cdot [w_1^0\dots w_n^0]
	\prod_{1\leq i<k\leq n} \frac{(1-s_i^{-1}s_kw_i^{-1}w_k)(1-w_i^{-1}w_k)}{(1-s_i^{-1}w_i^{-1}w_k)(1-s_kw_i^{-1}w_k)}\\
	&\qquad\cdot \prod_{j=1}^n\prod_{a=1}^t \frac{(-Qs_jw_j\xi_t^{-a+1},-s_j^{-1}w_j^{-1}\xi_t^{-a};Q)_{\infty}}
	{(-Qw_j\xi_t^{-a+1},-w_j^{-1}\xi_t^{-a};Q)_{\infty}}
	\cdot
	\prod_{i,k=1}^n \frac{(Qs_i^{-1}s_kw_i^{-1}w_k,Qw_i^{-1}w_k;Q)_{\infty}}
	{(Qs_i^{-1}w_i^{-1}w_k,Qs_kw_i^{-1}w_k;Q)_{\infty}}.
\end{align*}
Using the infinite product formula \eqref{eqn:def vartheta} for the theta function $\vartheta(z;Q)$,
the above equation can be represented by
\begin{align*}
	&\prod_{j=1}^n \frac{s_j^{-t/2}}
	{\vartheta(s_j;Q)}
	\cdot [w_1^0\dots w_n^0]
	\prod_{i<k}^n \frac{\vartheta(s_i^{-1}s_kw_i^{-1}w_k;Q)
		\vartheta(w_i^{-1}w_k;Q)}
	{\vartheta(s_i^{-1}w_i^{-1}w_k;Q)
		\vartheta(s_kw_i^{-1}w_k;Q)}
	\cdot \prod_{j=1}^n\prod_{a=1}^t
	\frac{\vartheta(-s_jw_j\xi_t^{a};Q)}
	{\vartheta(-w_j\xi_t^{a};Q)},
\end{align*}
which is the right hand side of equation \eqref{eqn:first eqn for F_t}.
\end{proof}

Based on Theorem \ref{prop:qZn lim Zn} and Proposition \ref{thm:qZn as det},
the second formula for the $n$-point function of $t$-core partitions can be stated as the following
\begin{cor}\label{cor:Ft as []det}
For the $n$-point function of $t$-core partitions $F_t(Q;s_1,\dots,s_n)$,
we have
\begin{align}\label{eqn:Ft as []det}
	\begin{split}
	F_t(Q;s_1,&\dots,s_n)
	=\frac{s_{[n]}^{-t/2}}
	{\Theta_3(-Q_2)^{n-1} \cdot \Theta_3(-Q_2s_{[n]}^{-1})}\\
	&\cdot [w_1^0\dots w_n^0]\ \prod_{j=1}^n\prod_{a=1}^t \frac{\vartheta(-s_jw_j\xi_t^a)}{\vartheta(-w_j\xi_t^a)}
	\cdot 
	\det \Big(\frac{\Theta_3(-Q_2s_i^{-1}w_i^{-1}w_j)}
	{\vartheta(s_iw_iw_j^{-1})} \Big)_{i,j=1}^n.
	\end{split}
\end{align}
in the region \eqref{eqn:region}.
\end{cor}
\begin{proof}
The second argument of $\Theta_3(\cdot,\cdot)$ is not always $Q$ in the following equations,
so we do not omit it in this proof.
Motivated by Theorem \ref{prop:qZn lim Zn} and Proposition \ref{thm:qZn as det},
we need to compute
\begin{align*}
	&\lim_{q\rightarrow1^-}
	([w_1^0\dots w_n^0]\det \Big(\frac{\Theta_3\big(Q_2s_i^{-1}w_i^{-1}w_j;Q_1Q\big)E(Q;Q_1,q;s_iw_i,w_j)}
	{(1-s_i^{-1}w_i^{-1}w_j)}\Big))|_{q\rightarrow\xi_t q,Q_1\rightarrow q^t}\\
	=&[w_1^0\dots w_n^0]\det \Big(\frac{\Theta_3\big(-Q_2s_i^{-1}w_i^{-1}w_j;Q\big) \cdot (Q;Q)_\infty^2}
	{(1-s_i^{-1}w_i^{-1}w_j)
		\cdot (Qs_i^{-1}w_i^{-1}w_j,Qs_iw_iw_j^{-1};Q)_{\infty}}\\
	&\qquad\qquad\qquad\cdot \prod_{a=1}^\infty
	\frac{(1+s_i^{-1}w_i^{-1}\xi_t^{-a})\cdot (-Qs_iw_i\xi_t^{-a+1},-Qs_i^{-1}w_i^{-1}\xi_t^{-a};Q)_{\infty}}
		{(1+w_j^{-1}\xi_t^{-a})\cdot
		(-Qw_j\xi_t^{-a+1},-Qw_j^{-1}\xi_t^{-a};Q)_{\infty}}\Big),
\end{align*}
where we have used $\Theta_3(z;e^{2\pi i}Q)=\Theta_3(-z;Q)$.
Then from the infinite product formula \eqref{eqn:def vartheta} for the theta function $\vartheta(z;Q)$,
the above equation is equal to
\begin{align*}
	[w_1^0\dots w_n^0]
	\det \Big(\frac{\Theta_3\big(-Q_2s_i^{-1}w_i^{-1}w_j;Q\big)}
	{\sqrt{s_i^{-1}w_i^{-1}w_j}\cdot \vartheta(s_iw_iw_j^{-1};Q)}
	\cdot \prod_{a=1}^t \frac{\sqrt{s_i^{-1}w_i^{-1}}\cdot \vartheta(-s_iw_i\xi_t^a;Q)}
		{\sqrt{w_j}\cdot \vartheta(-w_j\xi_t^a;Q)} \Big).
\end{align*}

Based on Theorem \ref{prop:qZn lim Zn} and comparing the above computations to formula \eqref{eqn:qZn as det} for the $q$-deformed $n$-point function $Z(Q;Q_1,q;s_1,\dots,s_n)$,
we have
\begin{align*}
	F_t(Q;s_1,\dots,s_n)
	=&\lim_{q\rightarrow1^-} Z(Q;Q_1,q;s_1,\dots,s_n)|_{q\rightarrow\xi_tq, Q_1\rightarrow q^t}\\
	=&\frac{s_{[n]}^{-t/2}}
	{\Theta_3(-Q_2;Q)^{n-1} \cdot \Theta_3(-Q_2s_{[n]}^{-1};Q)}\\
	&\cdot [w_1^0\dots w_n^0]\ 
	\det \Big(\frac{\Theta_3\big(-Q_2s_i^{-1}w_i^{-1}w_j;Q\big)}
	{\vartheta(s_iw_iw_j^{-1};Q)}
	\cdot \prod_{a=1}^t \frac{\vartheta(-s_iw_i\xi_t^a;Q)}{\vartheta(-w_j\xi_t^a;Q)} \Big),
\end{align*}
which is equivalent to equation \eqref{eqn:Ft as []det}.
\end{proof}

\subsection{Closed formula for the $n$-point function of $t$-core partitions}
\label{sec:closed formula}
In this subsection,
we provide a closed formula for the $n$-point function of $t$-core partitions based on the equation \eqref{eqn:Ft as []det}.

\begin{lem}\label{cor:special res+res=0}
	For any given $1\leq l\leq n-1$,
	we have
	\begin{align}\label{eqn:special res+res=0}
		\begin{split}
		0=&[w_{l}^0]\ \Big(\Res_{w_n=s_{l}w_{l}Q}+\Res_{w_n=s_n^{-1}w_{l}Q}\Big)\\
		&\qquad\qquad \frac{1}{w_n}\cdot \prod_{j=1}^n\prod_{a=1}^t \frac{\vartheta(-s_jw_j\xi_t^a)}{\vartheta(-w_j\xi_t^a)}
		\cdot 
		\det \Big(\frac{\Theta_3(-Q_2s_i^{-1}w_i^{-1}w_j)}
		{\vartheta(s_iw_iw_j^{-1})} \Big)_{1\leq i,j\leq n}.
		\end{split}
	\end{align}
\end{lem}
\begin{proof}
This is the special case $r=0$ of equation \eqref{lem:res+res=0}.
As it is used in the first step of the proof of Theorem \ref{thm:Ft closed formula},
we prove it directly here.

We first compute the residues in equation \eqref{eqn:special res+res=0}.
Note that the pole $w_n=s_{l}w_{l}Q$ comes from the $(l,n)$-entry of the matrix determinant,
and the pole $w_ns_n^{-1}w_{l}Q$ comes from the $(n,l)$-entry.
Thus, by the residue theorem,
we have
\begin{align}\label{eqn:sl1wl1Q first}
	\begin{split}
	&\Res_{w_n=s_{l}w_{l}Q}\ \frac{1}{w_n}\cdot \prod_{j=1}^n\prod_{a=1}^t \frac{\vartheta(-s_jw_j\xi_t^a)}{\vartheta(-w_j\xi_t^a)}
	\cdot 
	\det \Big(\frac{\Theta_3(-Q_2s_i^{-1}w_i^{-1}w_j)}
	{\vartheta(s_iw_iw_j^{-1})} \Big)_{1\leq i,j\leq n}\\
	=&\prod_{j=1}^{n-1}\prod_{a=1}^t \frac{\vartheta(-s_jw_j\xi_t^a)}{\vartheta(-w_j\xi_t^a)}
	\cdot \prod_{a=1}^t \frac{\vartheta(-s_ns_{l}w_{l}Q\xi_t^a)}{\vartheta(-s_{l}w_{l}Q\xi_t^a)}
	\cdot \det(B),
	\end{split}
\end{align}
where $B$ is the following $(n\times n) $ matrix 
\begin{align*}
	B=\left(\begin{array}{c:c}
		\Big(\frac{\Theta_3(-Q_2s_i^{-1}w_i^{-1}w_j)}
		{\vartheta(s_iw_iw_j^{-1})}\Big)_{i=1,...,l-1\atop j=1,...,n-1} & (0)_{i=1,...,l-1\atop j=n}\\
		\hdashline (0)_{i=l\atop j=1,...,n-1} & \Res_{w_n=s_{l}w_{l}Q}\, \frac{\Theta_3(-Q_2s_{l}^{-1}w_{l}^{-1}w_n)}
		{w_n\cdot \vartheta(s_{l}w_{l}w_n^{-1})}\\
		\hdashline \Big(\frac{\Theta_3(-Q_2s_i^{-1}w_i^{-1}w_j)}
		{\vartheta(s_iw_iw_j^{-1})}\Big)_{i=l+1,...,n-1\atop j=1,...,n-1} & (0)_{i=l+1,...,n-1\atop j=n}\\
		\hdashline \Big(\frac{\Theta_3(-Q_2s_n^{-1}s_{l}^{-1}w_{l}^{-1}Q^{-1}w_j)}
		{\vartheta(s_ns_{l}w_{l}Qw_j^{-1})}\Big)_{i=n\atop j=1,...,n-1} & (0)_{i=n\atop j=n}\\
	\end{array}\right).
\end{align*}
Since $s_{l}w_{l}Q$ is only a simple pole,
we have
$\Res_{w_n=s_{l}w_{l}Q}\, \frac{\Theta_3(-Q_2s_{l}^{-1}w_{l}^{-1}w_n)}
{w_n\cdot \vartheta(s_{l}w_{l}w_n^{-1})}
=\frac{\Theta_3(-Q_2Q)}
{Q^{-1/2}}$,
where we have used $\vartheta'(Q^{-1})=-Q^{1/2}$.
As a result,
by using the modular properties of $\Theta_3(z)$ and $\vartheta(z)$,
the residue of the corresponding function at $w_n=s_{l}w_{l}Q$,
i.e. the equation \eqref{eqn:sl1wl1Q first},
is equal to
\begin{align}\label{eqn:first pole}
	\begin{split}
		&s_n^{-t}\cdot
		\prod_{j=1\atop j\neq l}^{n-1}\prod_{a=1}^t \frac{\vartheta(-s_jw_j\xi_t^a)}{\vartheta(-w_j\xi_t^a)}
		\cdot \prod_{a=1}^t \frac{\vartheta(-s_ns_{l}w_{l}\xi_t^a)}{\vartheta(-w_{l}\xi_t^a)}\\
		&\qquad\qquad\quad\cdot (-1)^{n+l+1}\Theta_3(-Q_2)
		\cdot\det\left(\begin{array}{c}
			\Big(\frac{\Theta_3(-Q_2s_i^{-1}w_i^{-1}w_j)}
			{\vartheta(s_iw_iw_j^{-1})}\Big)_{i=1,\dots,l-1,l+1,\dots,n-1\atop j=1,...,n-1}\\
			\hdashline \Big(\frac{\Theta_3(-Q_2s_n^{-1}s_{l}^{-1}w_{l}^{-1}w_j)}
			{\vartheta(s_ns_{l}w_{l}w_j^{-1})}\Big)_{i=n\atop j=1,...,n-1}
		\end{array}\right).
	\end{split}
\end{align}
The residue at $w_n=s_n^{-1}w_{l}Q$ of the corresponding function can be computed similarly,
so we omit the details.
Just notice that in this case the pole appears in the $(n,l)$-th element of the matrix determinant and
$\Res_{w_n=s_{n}^{-1}w_{l}Q} \frac{\Theta_3(-Q_2s_{n}^{-1}w_{n}^{-1}w_{l})}
{w_n \cdot \vartheta(s_{n}w_{n}w_{l}^{-1})}
=\frac{\Theta_3(-Q_2Q^{-1})}
{-Q^{-1/2}}$.
So the final result is
\begin{align}\label{eqn:second pole}
	\begin{split}
	&\Res_{w_n=s_n^{-1}w_{l}Q}\ \frac{1}{w_n}\prod_{j=1}^n\prod_{a=1}^t \frac{\vartheta(-s_jw_j\xi_t^a)}{\vartheta(-w_j\xi_t^a)}
	\cdot 
	\det \Big(\frac{\Theta_3(-Q_2s_i^{-1}w_i^{-1}w_j)}
	{\vartheta(s_iw_iw_j^{-1})} \Big)_{1\leq i,j\leq n}\\
	=&s_n^{-t}
	\cdot \prod_{j=1\atop j\neq l}^{n-1}\prod_{a=1}^t \frac{\vartheta(-s_jw_j\xi_t^a)}{\vartheta(-w_j\xi_t^a)}
	\cdot \prod_{a=1}^t \frac{\vartheta(-w_{l}\xi_t^a)}{\vartheta(-s_n^{-1}w_{l}\xi_t^a)}
	\cdot (-1)^{n+l}\Theta_3(-Q_2)\\
	&\cdot\det\left(\begin{array}{c:c}
		\Big(\frac{\Theta_3(-Q_2s_i^{-1}w_i^{-1}w_j)}
		{\vartheta(s_iw_iw_j^{-1})}\Big)_{i=1,...,n-1\atop j=1,...,l-1,l+1,\dots,n-1}
		& \Big(\frac{\Theta_3(-Q_2s_i^{-1}w_i^{-1}s_n^{-1}w_{l})}
		{\vartheta(s_iw_is_nw_{l}^{-1})}\Big)_{i=1,...,n-1\atop j=n}
	\end{array}\right).
\end{split}
\end{align}
Observe that if we move the $l$-th column of the matrix determinant in equation \eqref{eqn:first pole} to the last column,
and in equation \eqref{eqn:second pole} we replace $w_{l}\rightarrow s_nw_{l}$ 
and move the $l$-th row to the last row,
the two results become opposites.
Notice that the change of variable $w_{l}\rightarrow s_nw_{l}$ does not affect the coefficient of $w_{l}^0$.
Hence,
the coefficient of $w_{l}^0$ in the sum of the two residues vanishes,
which proves equation \eqref{eqn:special res+res=0}.
\end{proof}

The main result of this subsection is the following closed formula for the $n$-point function of $t$-core partitions $F_t(Q;s_1,\dots,s_n)$.
This formula has been stated in Theorem \ref{thm:main formula}.
For convenience,
we restate it here with more details.
\begin{thm}\label{thm:Ft closed formula}
	The $n$-point function of $t$-core partitions $F_t(s_1,\dots,s_n)$ has the following closed formula
	\begin{align}\label{eqn:Ft closed formula}
		\begin{split}
		F_t(Q;s_1,\dots,&s_n)
		=\frac{1}
		{\prod_{j=1}^n(s_j^{t/2}-s_j^{-t/2}) \cdot \Theta_3(-Q_2s_{[n]}^{-1})}
		\cdot\sum_{k=1}^{n}
		\bigg(\frac{1}{\Theta_3(-Q_2)^{k-1}}\\
		&\cdot
		\sum_{\{\pi_1,\dots,\pi_k\}\vdash[n]}
		\sum_{l_1,\dots,l_k=1}^t
		\prod_{m=1}^k \frac{\prod_{a=1}^t\vartheta(s_{\pi_m}\xi_t^{a-l_m})}
		{\prod_{a=1 \atop a\neq l_m}^t\vartheta(\xi_t^{a-l_m})}
		\cdot \det(A_{l_1,\dots,l_k;i,j}^{\pi_1,\dots,\pi_k})_{i,j=1}^{k}\bigg),
		\end{split}
	\end{align}
	where $[n]=\{1,2,...,n\}$,
	the second sum runs over all set partitions $\pi=\{\pi_1,\dots,\pi_k\}$ of $[n]$ into $k$ blocks,
	i.e., each $\pi_i$ is a nonempty subset of $[n]$ such that $\sqcup_{i=1}^k \pi_k = [n]$,
	$s_{\pi_m}=\prod_{x\in\pi_m} s_{x}$,
	and the entry $A_{l_1,\dots,l_k;i,j}^{\pi_1,\dots,\pi_k}$ of the matrix determinant is given by
	\begin{align}\label{eqn:def A^{pi}_ij}
		A_{l_1,\dots,l_k;i,j}^{\pi_1,\dots,\pi_k}
		=\frac{\Theta_3(-Q_2s_{\pi_{i}}^{-1}\xi_t^{l_{i}-l_{j}})}
		{\vartheta(s_{\pi_{i}}\xi_t^{l_{j}-l_{i}})}.	\end{align}
\end{thm}
\begin{proof}
We start from the following formula
\begin{align}\label{eqn:Ft as []g}
	F_t(Q;s_1,\dots,s_n)
	=\frac{s_{[n]}^{-t/2}}
	{\Theta_3(-Q_2)^{n-1} \cdot \Theta_3(-Q_2s_{[n]}^{-1})}
	\cdot [w_1^0\dots w_n^0]\, g_0,
\end{align}
where the function $g_0$ is defined by
\begin{align}\label{eqn:def g0}
	g_0:=\prod_{j=1}^n\prod_{a=1}^t \frac{\vartheta(-s_jw_j\xi_t^a)}{\vartheta(-w_j\xi_t^a)}
	\cdot 
	\det \Big(\frac{\Theta_3(-Q_2s_i^{-1}w_i^{-1}w_j)}
	{\vartheta(s_iw_iw_j^{-1})} \Big)_{1\leq i,j\leq n}.
\end{align}
The equation \eqref{eqn:Ft as []g} is just a reformulation of that in Corollary \ref{cor:Ft as []det}.
It is known that considering the coefficient of $w_1^0\dots w_n^0$ in the function $g_0$ is equivalent to considering the following integral
\begin{align}\label{eqn:[] as int}
	[w_1^0\dots w_n^0]\, g_0
	=\int_{|w_1|=e^{c_1},\dots,|w_n|=e^{c_n}}
	g_0\cdot \frac{ dw_1\cdots dw_n}{(2\pi i)^n w_1\dots w_n},
\end{align}
where these constants $c_i, i=1,\dots,n$ satisfy  $$1<e^{c_n}<|s_ne^{c_n}|<\dots<e^{c_1}<|s_1e^{c_1}|< |Q^{-1}|$$
since we have restricted ourselves to the region \eqref{eqn:region}.
Below,
we compute the above integral \eqref{eqn:[] as int} by $n$ steps.
For each $r=1,2,...,n$ and at the $r$-th step,
we compute the integral with respect to the variable $w_{n+1-r}$.
We denote the result by $g_r$ after finishing the first $r$ steps,
i.e.,
\begin{align}\label{eqn:def gr}
	g_r:=[w_{n+1-r}^0w_{n+2-r}^0\dots w_{n}^0]\ g_0
	=\int_{|w_{n+1-r}|=e^{c_{n+1-r}},\dots,|w_n|=e^{c_n}}
	g_0\cdot \frac{ dw_{n+1-r}\cdots dw_n}{(2\pi i)^r w_{n+1-r}\dots w_n}.
\end{align}
We prove by induction on $r$ that,
after performing the first $r$ times integrals,
we have the following equality
\begin{align}\label{eqn:induction on r eqn}
	\begin{split}
	g_r
	=&\frac{1}{\prod_{j=n+1-r}^n(1-s_j^{-t})}
	\cdot \sum_{k=1}^{r} \bigg(\Theta_3(-Q_2)^{r-k}
	\cdot \sum_{\pi=\{\pi_1,\dots,\pi_k\}\vdash \{n+1-r,\dots,n\}}
	\sum_{l_1,\dots,l_k=1}^t\\
	&\qquad\prod_{j=1}^{n-r}\prod_{a=1}^t \frac{\vartheta(-s_jw_j\xi_t^a)}{\vartheta(-w_j\xi_t^a)}
	\cdot\prod_{m=1}^k \frac{\prod_{a=1}^t\vartheta(s_{\pi_m}\xi_t^{a-l_m})}
	{\prod_{a=1 \atop a\neq l_m}^t\vartheta(\xi_t^{a-l_m})}
	\cdot \det(A_{l_1,\dots,l_k;i,j}^{n-r;\pi_1,\dots,\pi_k})_{i,j=1}^{n-r+k}\bigg),
	\end{split}
\end{align}
where the second sum runs over all set partitions $\pi=\{\pi_1,\dots,\pi_k\}$ of the set $\{n+1-r,n+2-r,\dots,n\}$ into $k$ blocks,
and the entries $A_{l_1,\dots,l_k;i,j}^{n-r;\pi_1,\dots,\pi_k}$ of the matrix determinant are given by
\begin{align}\label{eqn:def A^{n-r,pi}_ij}
	A_{l_1,\dots,l_k;i,j}^{n-r;\pi_1,\dots,\pi_k}
	=\begin{cases}
		\frac{\Theta_3(-Q_2s_i^{-1}w_i^{-1}w_j)}
		{\vartheta(s_iw_iw_j^{-1})}, &\text{if}\ 1\leq i,j\leq n-r,\\
		\frac{\Theta_3(Q_2s_i^{-1}w_i^{-1}\xi_t^{-l_{j-n+r}})}
		{-\vartheta(-s_iw_i\xi_t^{l_{j-n+r}})}, &\text{if}\ 1\leq i\leq n-r, n-r+1\leq j\leq n-r+k,\\
		\frac{\Theta_3(Q_2s_{\pi_{i-n+r}}^{-1}\xi_t^{l_{i-n+r}}w_j)}
		{\vartheta(-s_{\pi_{i-n+r}}\xi_t^{-l_{i-n+r}}w_j^{-1})}, &\text{if}\ n-r+1\leq i\leq n-r+k, 1\leq j\leq n-r,\\
		\frac{\Theta_3(-Q_2s_{\pi_{i-n+r}}^{-1}\xi_t^{l_{i-n+r}-l_{j-n+r}})}
		{\vartheta(s_{\pi_{i-n+r}}\xi_t^{l_{j-n+r}-l_{i-n+r}})}, &\text{if}\ n-r+1\leq i, j\leq n-r+k.\\
	\end{cases}
\end{align}
Notice that,
if we simultaneous permute $\pi_1,\dots,\pi_k$ and $l_1,\dots,l_k$ by the same permutation,
the determinant $\det(A_{l_1,\dots,l_k;i,j}^{n-r;\pi_1,\dots,\pi_k})_{i,j=1}^{n-r+k}$ is  unchanged.
Moreover,
when $r=n$,
this $A_{l_1,\dots,l_k;i,j}^{n-r;\pi_1,\dots,\pi_k}$ is equal to $A_{l_1,\dots,l_k;i,j}^{\pi_1,\dots,\pi_k}$ defined in equation \eqref{eqn:def A^{pi}_ij},
and thus equations \eqref{eqn:Ft as []g}, \eqref{eqn:def gr} and \eqref{eqn:induction on r eqn} give the formula \eqref{eqn:Ft closed formula}.

Now,
we prove the formula \eqref{eqn:induction on r eqn} by induction on $r$.
We first consider the $r=1$ case,
i.e., we compute the integral $\int_{|w_n|=e^{c_n}} \frac{g_0dw_n}{2\pi iw_n}$.
By the modular properties of theta functions $\Theta_3(z)$ and $\vartheta(z)$ performed in the last subsection,
we have (see \cite{Eng21} and Lemma 4 in \cite{EO06})
\begin{align}\label{eqn:int Qc=int c}
	\int_{|w_n|=Qe^{c_n}} \frac{g_0 dw_n}{w_n}
	=\int_{|w_n|=e^{c_n}} (g_0|_{w_n \rightarrow Qw_n}) \cdot\frac{dw_n}{w_n}
	=s_n^{-t}\cdot \int_{|w_n|=e^{c_n}} \frac{g_0dw_n}{w_n}.
\end{align}
In the region $\{Qe^{c_n}<|w_n|<e^{c_n}\}$,
the only poles of the function $g_0/w_n$ comes from the possible zeros of the functions $\vartheta(-w_n\xi_t^{l_1}), l_1=1,\dots,t$ and $\vartheta(s_{i}w_{i}w_n^{-1}), \vartheta(s_nw_nw_{i}^{-1}), i=1,\dots,n-1$.
Consequently,
the poles are at $w_n=-\xi_t^{-l_1}, l_1=1,\dots,t$
and $w_n=s_iw_iQ, s_n^{-1}w_iQ, i=1,...,n$.
Then by the residue theorem,
we have
\begin{footnotesize}
\begin{align*}
	\int_{|w_n|=e^{c_n}} \frac{g_0dw_n}{2\pi iw_n}
	=\int_{|w_n|=Qe^{c_n}} \frac{g_0dw_n}{2\pi i w_n}
	+ \bigg(\sum_{l_1=1}^t \Res_{w_n=-\xi_t^{-l_1}}
	+\sum_{i=1}^{n-1}\Big(\Res_{w_n=s_iw_iQ}+\Res_{w_n=s_n^{-1}w_iQ}\Big) \bigg)\frac{g_0}{w_n}.
\end{align*}
\end{footnotesize}
Note that the poles $s_iw_i$ and $s_n^{-1}w_i, i=1,...,n-1$ of $g_0/w_n$ with respect to the variable $w_n$ are not in the region $\{Qe^{c_n}<|w_n|<e^{c_n}\}$,
since $|e^{c_n}|<|s_iw_i|$ and $|s_ne^{c_n}|<|w_i|$
by our initial choice of contours.
Then by equation \eqref{eqn:int Qc=int c},
we have
\begin{align*}
	\int_{|w_n|=e^{c_n}} \frac{g_0dw_n}{2\pi iw_n}
	=\frac{1}{1-s_n^{-t}}\bigg(\sum_{l_1=1}^t \Res_{w_n=-\xi_t^{-l_1}}
	+\sum_{i=1}^{n-1}\Big(\Res_{w_n=s_iw_iQ}+\Res_{w_n=s_n^{-1}w_iQ}\Big) \bigg)\frac{g_0}{w_n}.
\end{align*}
The integral thus reduces to residues.
Lemma \ref{cor:special res+res=0} shows that for each $i=1,...,n-1$, the sum of residues of the function $g_0/w_n$ at $w_n=s_iw_iQ$ and $w_n=s_n^{-1}w_iQ$ vanishes.
Thus,
we only need to compute the residues of the function $g_0/w_n$ at $w_n=-\xi_t^{-l_1}, l_1=1,...,t$.
The result is
\begin{align}\label{eqn:int wn=result}
	\begin{split}
	g_1=&\int_{|w_n|=e^{c_n}}\ \frac{g_0dw_n}{2\pi iw_n}
	=\frac{1}{1-s_n^{-t}}
	\cdot \sum_{l_1=1}^t \Res_{w_n=-\xi_t^{-l_1}} \frac{g_0}{w_n}\\
	=&\frac{1}{1-s_n^{-t}}
	\cdot \sum_{l_1=1}^t\bigg(
	\prod_{j=1}^{n-1}\prod_{a=1}^t \frac{\vartheta(-s_jw_j\xi_t^a;Q)}{\vartheta(-w_j\xi_t^a;Q)}
	\cdot  \frac{\prod_{a=1}^t \vartheta(s_n\xi_t^{a-l_1};Q)}
		{\prod_{a=1 \atop a\neq l_1}^t \vartheta(\xi_t^{a-l_1};Q)}\\
	&\qquad\qquad\qquad\quad\cdot 
	\det \Big(\frac{\Theta_3(-Q_2s_i^{-1}w_i^{-1}w_j)}
	{\vartheta(s_iw_iw_j^{-1})}|_{w_n\rightarrow-\xi_t^{-l_1}} \Big)_{1\leq i,j\leq n}\bigg).
	\end{split}
\end{align}
Note that the substitution $w_n\rightarrow-\xi_t^{-l_1}$ only affects the $n$-row and $n$-th column of the matrix determinant.
Hence
the entries $\frac{\Theta_3(-Q_2s_i^{-1}w_i^{-1}w_j)}
{\vartheta(s_iw_iw_j^{-1})}|_{w_n\rightarrow-\xi_t^{l_1}}$ in the matrix determinant coincides with $A^{n-1;\{n\}}_{l_1;i,j}$
defined in equation \eqref{eqn:def A^{n-r,pi}_ij} for all $1\leq i,j\leq n$.
Since in this case the only set partition of $\{n\}$ is $\pi=\{\pi_1\}=\{\{n\}\}$, we have proved the equation \eqref{eqn:induction on r eqn} for the $r=1$ case.

Now assume that for some fixed $r$ with $2\leq r+1\leq n$,
we have already performed the first $r$ steps,
i.e.,
we assume that we have derived the formula \eqref{eqn:induction on r eqn} for the function $g_r$.
Next,
we perform the $(r+1)$-th step,
i.e., extracting the coefficient of $w_{n-r}^0$ in the function $g_r$.
We restrict us to each fixed summand in equation \eqref{eqn:induction on r eqn} of $g_r$,
i.e.,
we consider a fixed set partition $\pi=\{\pi_1,\dots,\pi_k\}$ and fixed integers $l_1,\dots,l_k$.
Notice that the poles of the function $g_r$ with respect to the variable $w_{n-r}$ in the region $\{Qe^{c_{n-r}}<|w_{n-r}|<e^{c_{n-r}}\}$ are three types as follows
\begin{align*}
	\begin{cases}
		-\xi_t^{-l},\quad l=1,...,t;\\
		s_iw_iQ,\, s_{n-r}^{-1}w_iQ,\quad i=1,...,{n-r-1};\\
		-s_{\pi_{m'}}\xi_t^{-l_{m'}},\quad {m'}=1,...,k.
	\end{cases}
\end{align*}
We remark that the points $-s_{n-r}^{-1}\xi_t^{-l_{m'}},\, {m'}=1,...,k$ are not poles of the function $g_r$ in the variable $w_{n-r}$.
Then by the modularity of the theta functions $\Theta_3(z), \vartheta(z)$ and a similar discussion as above,
we have
\begin{align}
	\begin{split}
		[w_{n-r}^0]g_r
		=&\frac{1}{1-s_{n-r}^{-t}}
		\cdot \bigg(\sum_{l=1}^t \Res_{w_{n-r}=-\xi_t^{-l}}
		+\sum_{{m'}=1}^k \Res_{w_{n-r}=-s_{\pi_{m'}}\xi_t^{-l_{m'}}}\\
		&\qquad\qquad\qquad+\sum_{i=1}^{n-r-1}\Big(\Res_{w_{n-r}=s_iw_iQ}+\Res_{w_{n-r}=s_{n-r}^{-1}w_iQ}\Big)\bigg)\frac{g_r}{w_{n-r}}.
	\end{split}
\end{align}
In Lemma \ref{lem:res+res=0},
we will prove that the sum of residues of the function $g_r/w_{n-r}$ at $w_{n-r}=s_iw_iQ$ and $w_{n-r}=s_{n-r}^{-1}w_iQ$ is equal to zero for any $i=1,...,n-r-1$.
Thus,
we only need to consider the residues of the function $g_r/w_{n-r}$ at $w_{n-r}=-\xi_t^{-l}, l=1,...,t$ and 
$w_{n-r}=-s_{\pi_{m'}}\xi_t^{-l_{m'}}, {m'}=1,...,k$.
Moreover,
since taking the residue is linear,
we can fix $\pi=\{\pi_1,\dots,\pi_k\}, l_1,\dots,l_k$ and focus on the terms involving $w_{n-r}$.
That is to say,
we consider the following function
\begin{align*}
	\tilde{g_r}^{\pi,(l_1,\dots,l_k)}
	:=\prod_{a=1}^t \frac{\vartheta(-s_{n-r}w_{n-r}\xi_t^a)}{\vartheta(-w_{n-r}\xi_t^a)}
	\cdot\prod_{m=1}^k \frac{\prod_{a=1}^t\vartheta(s_{\pi_m}\xi_t^{a-l_m})}
	{\prod_{a=1 \atop a\neq l_m}^t\vartheta(\xi_t^{a-l_m})}
	\cdot \det(A_{i,j}^{n-r;\pi_1,\dots,\pi_k})_{i,j=1}^{n-r+k},
\end{align*}
and compute the residues of $\tilde{g_r}^{\pi,(l_1,\dots,l_k)}/w_{n-r}$.
Then,
for any fixed $l=1,...,t$ and the pole at $w_{n-r}=-\xi_t^{-l}$,
we have
\begin{align}\label{eqn:res wn-r xi}
	\begin{split}
	\Res_{w_{n-r}=-\xi_t^{-l}}\frac{\tilde{g_r}^{\pi,(l_1,\dots,l_k)}}{w_{n-r}}
	&=\frac{\prod_{a=1}^t\vartheta(s_{n-r}\xi_t^{a-l})}
	{\prod_{a=1 \atop a\neq l}^t\vartheta(\xi_t^{a-l})}
	\cdot\prod_{m=1}^k \frac{\prod_{a=1}^t\vartheta(s_{\pi_m}\xi_t^{a-l_m})}
	{\prod_{a=1 \atop a\neq l_m}^t\vartheta(\xi_t^{a-l_m})}\\
	&\cdot \det\Big(A_{l_1,\dots,l_k;i,j}^{n-r;\pi_1,\dots,\pi_k}|_{w_{n-r}\rightarrow-\xi_t^{-l}}\Big)_{i,j=1}^{n-r+k}.
	\end{split}
\end{align}
Notice that for any $1\leq i,j\leq n-r+k$,
we have
\begin{align*}
	A_{l_1,\dots,l_k;i,j}^{n-r;\pi_1,\dots,\pi_k}|_{w_{n-r}\rightarrow-\xi_t^{-l}}
	=A_{l,l_1,\dots,l_k;i,j}^{n-r-1;\{n-r\},\pi_1,\dots,\pi_k}.
\end{align*}
Hence
the contribution in \eqref{eqn:res wn-r xi} corresponds to the term in $g_{r+1}$ where
the set partition of $\{n-r,n-r+1,\dots,n\}$ is given by $\{\pi_1,,\dots,\pi_k,\{n-r\}\}$,
(i.e., the number $n-r$ is a singleton)
and the associated label is $l_{k+1}=l$.
Next,
for any fixed ${m'}=1,...,k$ and the pole at $w_{n-r}=s_{\pi_{m'}}\xi_t^{-l_{m'}}$,
we have
\begin{align}\label{eqn:res wn-r sxi}
	\begin{split}
	&\Res_{w_{n-r}=-s_{\pi_{m'}}\xi_t^{-l_{m'}}} \frac{\tilde{g_r}^{\pi,(l_1,\dots,l_k)}}{w_{n-r}}\\
	=&\prod_{a=1}^t \frac{\vartheta(s_{n-r}s_{\pi_{m'}}\xi_t^{a-l_{m'}})}{\vartheta(s_{\pi_{m'}}\xi_t^{a-l_{m'}})}
	\cdot\prod_{m=1}^k \frac{\prod_{a=1}^t\vartheta(s_{\pi_m}\xi_t^{a-l_m})}
	{\prod_{a=1 \atop a\neq l_m}^t\vartheta(\xi_t^{a-l_m})}
	\cdot \Res_{w_{n-r}=-s_{\pi_{m'}}\xi_t^{-l_{m'}}} \det(A_{l_1,\dots,l_k;i,j}^{n-r;\pi_1,\dots,\pi_k})_{i,j=1}^{n-r+k}\\
	=&\frac{\prod_{a=1}^t\vartheta(s_{n-r}s_{\pi_{m'}}\xi_t^{a-l_{m'}})}
		{\prod_{a=1 \atop a\neq l_{m'}}^t \vartheta(\xi_t^{a-l_{m'}})}
	\cdot\prod_{m=1\atop m\neq {m'}}^k \frac{\prod_{a=1}^t\vartheta(s_{\pi_m}\xi_t^{a-l_m})}
	{\prod_{a=1 \atop a\neq l_m}^t\vartheta(\xi_t^{a-l_m})}
	\cdot \det(B_{i,j})_{i,j=1}^{n-r+k},
	\end{split}
\end{align}
where by using $\Res_{w_{n-r}=-s_{\pi_{m'}}\xi_t^{-l_{m'}}} \frac{\Theta_3(Q_2s_{\pi_{m'}}^{-1}\xi_t^{l_{m'}}w_{n-r})}
{w_{n-r}\cdot \vartheta(-s_{\pi_{m'}}\xi_t^{-l_{m'}}w_{n-r}^{-1})}
=-\Theta_3(-Q_2)$,
the matrix determinant $\det(B_{i,j})_{i,j=1}^{n-r+k}$ is given by
\begin{align}\label{eqn:B matrix}
	\begin{split}
	&\det\begin{footnotesize}
		\begin{array}{cc}
			\left(\begin{array}{c:c:c}
				\frac{\Theta_3(-Q_2s_i^{-1}w_i^{-1}w_j)}
				{\vartheta(s_iw_iw_j^{-1})}&
				0&
				\frac{\Theta_3(Q_2s_i^{-1}w_i^{-1}\xi_t^{-l_{j-n+r}})}
				{-\vartheta(-s_iw_i\xi_t^{l_{j-n+r}})}\\
				\hdashline \frac{\Theta_3(Q_2s_{n-r}^{-1}s_{\pi_{m'}}^{-1}\xi_t^{l_{m'}}w_j)}
				{\vartheta(-s_{n-r}s_{\pi_{m'}}\xi_t^{-l_{m'}}w_j^{-1})}&
				0&
				\frac{\Theta_3(-Q_2s_i^{-1}s_{\pi_{m'}}^{-1}\xi_t^{l_{m'}-l_{j-n+r}})}
				{\vartheta(s_{n-r}s_{\pi_{m'}}\xi_t^{l_{j-n+r}-l_{m'}})}\\
				\hdashline \frac{\Theta_3(Q_2s_{\pi_{i-n+r}}^{-1}\xi_t^{l_{i-n+r}}w_j)}
				{\vartheta(-s_{\pi_{i-n+r}}\xi_t^{-l_{i-n+r}}w_j^{-1})}&
				0&
				\frac{\Theta_3(-Q_2s_{\pi_{i-n+r}}^{-1}\xi_t^{l_{i-n+r}-l_{j-n+r}})}
				{\vartheta(s_{\pi_{i-n+r}}\xi_t^{l_{j-n+r}-l_{i-n+r}})}\\
				\hdashline0&
				-\Theta_3(-Q_2)&
				0\\
				\hdashline \underbrace{\textstyle \frac{\Theta_3(Q_2s_{\pi_{i-n+r}}^{-1}\xi_t^{l_{i-n+r}}w_j)}
					{\vartheta(-s_{\pi_{i-n+r}}\xi_t^{-l_{i-n+r}}w_j^{-1})}} &
				\underbrace{\vphantom{\frac{\Big(change\Big)}{\big(height\big)}}0}&
				\underbrace{\textstyle \frac{\Theta_3(-Q_2s_{\pi_{i-n+r}}^{-1}\xi_t^{l_{i-n+r}-l_{j-n+r}})}
					{\vartheta(s_{\pi_{i-n+r}}\xi_t^{l_{j-n+r}-l_{i-n+r}})}}
			\end{array}\right)
			&\begin{array}{l}
				\Big\}i={1,\dots,\atop n-r-1}\\[10pt]
				\Big\}i=n-r\\[10pt]
				\Big\}i={n-r+1,\dots,\atop m'+n-r-1}\\[10pt]
				\big\}i=m'+n-r\\[12pt]
				\Big\}i={m'+n-r+1,\dots,\atop n-r+k}\\[5.5pt]
			\end{array}\\
			\begin{array}{ccc}
				\,\quad j=1,\dots,n-r-1&\,\,\,\quad j=n-r&\,\, j=n-r+1,\dots,n-r+k\,\,
			\end{array}
		\end{array}\end{footnotesize}\\
	=&\begin{footnotesize} {\textstyle(-1)^{m'+1}\Theta_3(-Q_2)}
		\begin{array}{cc}
			\left|\begin{array}{c:c}
				\frac{\Theta_3(-Q_2s_i^{-1}w_i^{-1}w_j)}
				{\vartheta(s_iw_iw_j^{-1})}&
				\frac{\Theta_3(Q_2s_i^{-1}w_i^{-1}\xi_t^{-l_{j-n+r}})}
				{-\vartheta(-s_iw_i\xi_t^{l_{j-n+r}})}\\
				\hdashline \frac{\Theta_3(Q_2s_{n-r}^{-1}s_{\pi_{m'}}^{-1}\xi_t^{l_{m'}}w_j)}
				{\vartheta(-s_{n-r}s_{\pi_{m'}}\xi_t^{-l_{m'}}w_j^{-1})}&
				\frac{\Theta_3(-Q_2s_i^{-1}s_{\pi_{m'}}^{-1}\xi_t^{l_{m'}-l_{j-n+r}})}
				{\vartheta(s_{n-r}s_{\pi_{m'}}\xi_t^{l_{j-n+r}-l_{m'}})}\\
				\hdashline \underbrace{\textstyle \frac{\Theta_3(Q_2s_{\pi_{i-n+r}}^{-1}\xi_t^{l_{i-n+r}}w_j)}
					{\vartheta(-s_{\pi_{i-n+r}}\xi_t^{-l_{i-n+r}}w_j^{-1})}} &
				\underbrace{\textstyle \frac{\Theta_3(-Q_2s_{\pi_{i-n+r}}^{-1}\xi_t^{l_{i-n+r}-l_{j-n+r}})}
					{\vartheta(s_{\pi_{i-n+r}}\xi_t^{l_{j-n+r}-l_{i-n+r}})}}
			\end{array}\right|
			&\begin{array}{l}
				\Big\}i={1,\dots,\atop n-r-1}\\[12pt]
				\Big\}i=n-r\\[15pt]
				\Big\}i={n-r+1,\dots,m'+n-r-1\atop m'+n-r+1,\dots, n-r+k}\\[10pt]
			\end{array}\\
			\begin{array}{cc}
				\,j=1,\dots,n-r-1\qquad& j=n-r+1,\dots,n-r+k.
			\end{array}
		\end{array}
	\end{footnotesize}
\end{split}
\end{align}
Note that
after moving the $(n-r)$-th row of the matrix determinant in the last line of above equation to its $(n-r+m'-1)$-th row,
the $(i,j)$-element of the result becomes $A_{l_1,\dots,l_k;i,j}^{n-r-1;\pi_1,\dots,\pi_{m'}\cup\{n-r\},\dots,\pi_k}$ for all $1\leq i,j\leq n-r+k-1$.
That is to say,
\begin{align}\label{eqn:B=A}
	\det(B_{i,j})_{i,j=1}^{n-r+k}
	=\Theta_3(-Q_2)
	\cdot \det\big(A_{l_1,\dots,l_k;i,j}^{n-r-1;\pi_1,\dots,\pi_{m'}\cup\{n-r\},\dots,\pi_k}\big)_{i,j=1}^{n-r+k-1}.
\end{align}
As a consequence,
combining with equation \eqref{eqn:B=A},
the contribution in equation \eqref{eqn:res wn-r sxi} corresponds to the term in $g_{r+1}$ where
the set partition of $\{n-r,n-r+1,\dots,n\}$ is given by $\{\pi_1,,\dots,\pi_{m'}\cup\{n-r\},\dots,\pi_k\}$
(i.e., the number $n-r$ is merged into $\pi_{m'}$).

Every set partition of $\{n-r,n-r+1,\dots,n\}$ is obtained from a set partition of $\{n-r+1,\dots,n\}$ either by adding $\{n-r\}$ as a singleton or by inserting the number $n-r$ into an existing block.
Thus by previous discussions in equation \eqref{eqn:res wn-r xi} and equation \eqref{eqn:res wn-r sxi},
the $(r+1)$-step is finished and the formula \eqref{eqn:induction on r eqn} is proved.
\end{proof}

The following lemma is used in the $(r+1)$-th step of proving Theorem \ref{thm:Ft closed formula}.
\begin{lem}\label{lem:res+res=0}
	For any given $0\leq r \leq n-1$ and $1\leq l\leq n-r-1$,
	we have
	\begin{align}\label{eqn:res+res=0}
		[w_{l}^0]\ 
		\Big(\Res_{w_{n-r}=s_lw_lQ}
		+\Res_{w_{n-r}=s_{n-r}^{-1}w_lQ}\Big) \frac{g_{r}}{w_{n-r}}=0,
	\end{align}
	where $g_r$ is the function given in equation \eqref{eqn:induction on r eqn}.
\end{lem}
\begin{proof}
	This result follows by direct computations.
	We first compute the residues in the left hand side of equation \eqref{eqn:res+res=0}.
	Since residues are linear,
	we can consider fixed $k$, $\pi=\{\pi_1,...,\pi_k\}$, $l_1,...,l_k$ and focus on the terms involving $w_{n-r}$ and $w_l$.
	That is to say, instead of using the function $g_r$ given in equation \eqref{eqn:induction on r eqn},
	we can consider the following function
	\begin{align*}
		\tilde{g_r}^{\pi,(l_1,\dots,l_k),l}
		:=\prod_{a=1}^t \frac{\vartheta(-s_lw_l\xi_t^a)}{\vartheta(-w_l\xi_t^a)}
		\cdot \prod_{a=1}^t \frac{\vartheta(-s_{n-r}w_{n-r}\xi_t^a)}{\vartheta(-w_{n-r}\xi_t^a)}
		\cdot \det(A_{i,j}^{n-r;\pi_1,\dots,\pi_k})_{i,j=1}^{n-r+k}
	\end{align*}
	and only need to prove
	\begin{align}\label{eqn:res+res=0 tildeg}
		[w_{l}^0]\ 
		\Big(\Res_{w_{n-r}=s_lw_lQ}
		+\Res_{w_{n-r}=s_{n-r}^{-1}w_lQ}\Big) \frac{\tilde{g_r}^{\pi,(l_1,\dots,l_k),l}}{w_{n-r}}=0.
	\end{align}
	Then,
	for the pole at $w_{n-r}=s_lw_lQ$,
	we have
	\begin{align}
		\label{eqn:slwlQresult}
		\begin{split}
		&\Res_{w_{n-r}=s_lw_lQ}\, \frac{\tilde{g_r}^{\pi,(l_1,\dots,l_k),l}}{w_{n-r}}
		=\prod_{a=1}^t \frac{\vartheta(-s_lw_l\xi_t^a)}{\vartheta(-w_l\xi_t^a)}
		\cdot \prod_{a=1}^t \frac{\vartheta(-s_{n-r}s_lw_lQ\xi_t^a)}{\vartheta(-s_lw_lQ\xi_t^a)}\\
		&\cdot \det\begin{footnotesize}
		\begin{array}{cc}
		\left(\begin{array}{c:c:c}
		\frac{\Theta_3(-Q_2s_i^{-1}w_i^{-1}w_j)}
		{\vartheta(s_iw_iw_j^{-1})}&
		0&
		\frac{\Theta_3(Q_2s_i^{-1}w_i^{-1}\xi_t^{-l_{j-n+r}})}
		{-\vartheta(-s_iw_i\xi_t^{l_{j-n+r}})}\\
		\hdashline0&
		\frac{\Theta_3(-Q_2Q)}{Q^{-1/2}}&
		0\\
		\hdashline\frac{\Theta_3(-Q_2s_i^{-1}w_i^{-1}w_j)}
		{\vartheta(s_iw_iw_j^{-1})}&
		0&
		\frac{\Theta_3(Q_2s_i^{-1}w_i^{-1}\xi_t^{-l_{j-n+r}})}
		{-\vartheta(-s_iw_i\xi_t^{l_{j-n+r}})}\\
		\hdashline \frac{\Theta_3(-Q_2s_{n-r}^{-1}s_l^{-1}w_l^{-1}Q^{-1}w_j)}
		{\vartheta(s_{n-r}s_lw_lQw_j^{-1})}&
		0 &
		\frac{\Theta_3(Q_2s_{n-r}^{-1}s_l^{-1}w_l^{-1}Q^{-1}\xi_t^{-l_{j-n+r}})}
		{-\vartheta(-s_{n-r}s_lw_lQ\xi_t^{l_{j-n+r}})}\\
		\hdashline \underbrace{\textstyle \frac{\Theta_3(Q_2s_{\pi_{i-n+r}}^{-1}\xi_t^{l_{i-n+r}}w_j)}
		{\vartheta(-s_{\pi_{i-n+r}}\xi_t^{-l_{i-n+r}}w_j^{-1})}} &
		\underbrace{\vphantom{\frac{\Big(change\Big)}{\big(height\big)}}0}&
		\underbrace{\textstyle \frac{\Theta_3(-Q_2s_{\pi_{i-n+r}}^{-1}\xi_t^{l_{i-n+r}-l_{j-n+r}})}
		{\vartheta(s_{\pi_{i-n+r}}\xi_t^{l_{j-n+r}-l_{i-n+r}})}}
		\end{array}\right)
		&\begin{array}{l}
			\Big\}i=1,\dots,l-1\\[8pt]
			\big\}i=l\\[3.5pt]
			\Big\}i={l+1,\dots,\atop n-r-1}\\[5.5pt]
			\Big\}i=n-r\\[14.5pt]
			\Bigg\}i={n-r+1,\dots,\atop n-r+k}\\[8.5pt]
		\end{array}\\
		\begin{array}{ccc}
			\,\quad j=1,\dots,n-r-1&\,\,\,\quad j=n-r&\, j=n-r+1,\dots,n-r+k
		\end{array}
		\end{array}
		\end{footnotesize}\\
		&=s_{n-r}^{-t}
		\cdot \prod_{a=1}^t \frac{\vartheta(-s_{n-r}s_lw_lQ\xi_t^a)}{\vartheta(-w_lQ\xi_t^a)}
		\cdot (-1)^{n-r+l+1}\Theta_3(-Q_2)\\
		&\qquad\quad\cdot \det\begin{footnotesize}
		\begin{array}{cc}
			\left(\begin{array}{c:c}
				\frac{\Theta_3(-Q_2s_i^{-1}w_i^{-1}w_j)}
				{\vartheta(s_iw_iw_j^{-1})}&
				\frac{\Theta_3(Q_2s_i^{-1}w_i^{-1}\xi_t^{-l_{j-n+r}})}
				{-\vartheta(-s_iw_i\xi_t^{l_{j-n+r}})}\\
				\hdashline \frac{\Theta_3(-Q_2s_{n-r}^{-1}s_l^{-1}w_l^{-1}w_j)}
				{\vartheta(s_{n-r}s_lw_lw_j^{-1})}&
				\frac{\Theta_3(Q_2s_{n-r}^{-1}s_l^{-1}w_l^{-1}\xi_t^{-l_{j-n+r}})}
				{-\vartheta(-s_{n-r}s_lw_l\xi_t^{l_{j-n+r}})}\\
				\hdashline \underbrace{\textstyle \frac{\Theta_3(Q_2s_{\pi_{i-n+r}}^{-1}\xi_t^{l_{i-n+r}}w_j)}
					{\vartheta(-s_{\pi_{i-n+r}}\xi_t^{-l_{i-n+r}}w_j^{-1})}} &
				\underbrace{\textstyle \frac{\Theta_3(-Q_2s_{\pi_{i-n+r}}^{-1}\xi_t^{l_{i-n+r}-l_{j-n+r}})}
					{\vartheta(s_{\pi_{i-n+r}}\xi_t^{l_{j-n+r}-l_{i-n+r}})}}
			\end{array}\right)
			&\begin{array}{l}
				\Big\}i={1,\dots,l-1, \atop l+1,\dots,n-r-1}\\[8pt]
				\Big\}i=n-r\\[14.5pt]
				\Bigg\}i={n-r+1,\dots,\atop n-r+k}\\[8.5pt]
			\end{array}\\
			\begin{array}{cc}
				\,\quad j=1,\dots,n-r-1&\, j=n-r+1,\dots,n-r+k.
			\end{array}
		\end{array}
		\end{footnotesize}
		\end{split}
	\end{align}
	Notice that in the first equal sign we have used $\Res_{w_{n-r}=s_lw_lQ}\frac{\Theta_3(-Q_2s_l^{-1}w_l^{-1}w_{n-r})}
	{\vartheta(s_lw_lw_{n-r}^{-1})}=\frac{\Theta_3(-Q_2Q)}{Q^{-1/2}}$,
	and in the second equal sign we have used the modular properties of theta functions.
	
	For the pole at $w_{n-r}=s_{n-r}^{-1}w_lQ$,
	the computation of taking residue is similar to the previous one,
	so we omit the details and only write down the result as
	\begin{align}\label{eqn:sn-r-1wlQresult}
		\begin{split}
		&\Res_{w_{n-r}=s_{n-r}^{-1}w_lQ}\, \frac{\tilde{g_r}^{\pi,(l_1,\dots,l_k),l}}{w_{n-r}}
		=s_{n-r}^{-t}
		\cdot \prod_{a=1}^t \frac{\vartheta(-s_lw_lQ\xi_t^a)}{\vartheta(-s_{n-r}^{-1}w_lQ\xi_t^a)}
		\cdot (-1)^{n-r+l}\Theta_3(-Q_2)\cdot\\
		&\begin{scriptsize}
			\begin{array}{cc}
				\left|\begin{array}{c:c:c}
					\frac{\Theta_3(-Q_2s_i^{-1}w_i^{-1}w_j)}
					{\vartheta(s_iw_iw_j^{-1})}&
					\frac{\Theta_3(-Q_2s_i^{-1}w_i^{-1}s_{n-r}^{-1}w_l)}
					{\vartheta(s_iw_is_{n-r}w_l^{-1})}&
					\frac{\Theta_3(Q_2s_i^{-1}w_i^{-1}\xi_t^{-l_{j-n+r}})}
					{-\vartheta(-s_iw_i\xi_t^{l_{j-n+r}})}\\
					\hdashline \underbrace{\textstyle \frac{\Theta_3(Q_2s_{\pi_{i-n+r}}^{-1}\xi_t^{l_{i-n+r}}w_j)}
						{\vartheta(-s_{\pi_{i-n+r}}\xi_t^{-l_{i-n+r}}w_j^{-1})}} &
					\underbrace{\textstyle \frac{\Theta_3(Q_2s_{\pi_{i-n+r}}^{-1}\xi_t^{l_{i-n+r}}s_{n-r}^{-1}w_l)}
						{\vartheta(-s_{\pi_{i-n+r}}\xi_t^{-l_{i-n+r}}s_{n-r}w_l^{-1})}}&
					\underbrace{\textstyle \frac{\Theta_3(-Q_2s_{\pi_{i-n+r}}^{-1}\xi_t^{l_{i-n+r}-l_{j-n+r}})}
						{\vartheta(s_{\pi_{i-n+r}}\xi_t^{l_{j-n+r}-l_{i-n+r}})}}
				\end{array}\right|
				&\begin{array}{l}
					\Big\}i={1,\dots,\atop n-r-1}\\[5.5pt]
					\Bigg\}i={n-r+1,\dots,\atop n-r+k}\\[8.5pt]
				\end{array}\\
				\begin{array}{cccc}
					\,\quad j={1,\dots,l-1 \atop l+1,\dots,n-r-1}&\,\qquad\qquad\qquad j=n-r&\, \qquad\qquad\qquad
					j=n-r+1,\dots,n-r+k,
				\end{array}
			\end{array}
		\end{scriptsize}
		\end{split}
	\end{align}
	where we have used $\Res_{w_{n-r}=s_{n-r}^{-1}w_lQ}\frac{\Theta_3(-Q_2s_{n-r}^{-1}w_{n-r}^{-1}w_l)}
	{\vartheta(s_{n-r}w_{n-r}w_{l}^{-1})}=\frac{\Theta_3(-Q_2Q)}{-Q^{-1/2}}$.
	Observation that
	if we deal with the matrix determinant in equation \eqref{eqn:slwlQresult} by moving its $l$-th column to its $(n-r-1)$-th column,
	and at the same time we deal with the matrix determinant in equation \eqref{eqn:sn-r-1wlQresult} by letting $w_{l}\rightarrow s_{n-r}w_{l}$
	and moving its $l$-row to the $(n-r-1)$-th row,
	the results become opposites.
	Because the change of variable $w_{l}\rightarrow s_{n-r}w_{l}$ does not affect the coefficient of $w_{l}^0$,
	the coefficient of $w_{l}^0$ in the sum of these two residues vanishes, which proves  equation \eqref{eqn:res+res=0}.
\end{proof}

As noted in the introduction,
both sides of equation \eqref{eqn:Ft closed formula} are actually independent of the variable $Q_2$.
Hence
we may set $Q_2$ to any convenient value to obtain a new formula for the $n$-point function $F_t(Q;s_1,\dots,s_n)$.
For examples,
for any given $r$ such that $1\leq r <n$,
one can set $Q_2=Q^{-1/2}\cdot\prod_{i=1}^r s_i$ to obtain the following formula,
which looks like simpler since it only involves the theta function $\vartheta(z)$.
\begin{cor}\label{cor:Ft closed formula-Q2}
	For any integer $r$ such that $1\leq r <n$,
	we have
	\begin{align}\label{eqn:Ft closed formula-Q_2}
		\begin{split}
			F_t(Q;s_1,&\dots,s_n)
			=\frac{1}{\prod_{j=1}^n(s_j^{t/2}-s_j^{-t/2})
				\cdot \vartheta(s_{[n]\setminus[r]}^{-1})}
			\cdot\sum_{k=1}^{n}
			\bigg(\frac{1}{\vartheta(s_{[r]})^{k-1}}\\
			&\cdot
			\sum_{\{\pi_1,\dots,\pi_k\}\vdash[n]}
			\sum_{l_1,\dots,l_k=1}^t
			\prod_{m=1}^k \frac{\prod_{a=1}^t\vartheta(s_{\pi_m}\xi_t^{a-l_m})}
			{\prod_{a=1 \atop a\neq l_m}^t\vartheta(\xi_t^{a-l_m})}
			\cdot \det\Big(\frac{\vartheta(s_{[r]}s_{\pi_{i}}^{-1}\xi_t^{l_{i}-l_{j}})}
			{\vartheta(s_{\pi_{i}}\xi_t^{l_{j}-l_{i}})}\Big)_{i,j=1}^{k}\bigg).
		\end{split}
	\end{align}
\end{cor}
\begin{proof}
	This corollary can be directly derived from Theorem \ref{thm:Ft closed formula} by setting $Q_2\rightarrow Q^{-1/2}\cdot s_{[r]}$.
	Notice that
	\begin{align*}
		\Theta_3(-Q^{-1/2}z)=j(z)=j'(1)\cdot z^{1/2}\theta(z).
	\end{align*}
	Then we have
	\begin{align*}
		\frac{\Theta_3(-Q_2)}{\Theta_3(-Q_2s_{[n]}^{-1})}\Big|_{Q_2\rightarrow Q^{-1/2}\cdot s_{[r]}}
		=s_{[n]}^{1/2}\cdot
		\frac{\vartheta(s_{[r]})}{\vartheta(s_{[n]\setminus[r]}^{-1})}
	\end{align*}
	and
	\begin{align*}
		\frac{\Theta_3(-Q_2s_{\pi_{i}}^{-1}\xi_t^{l_{i}-l_{j}})}{\Theta_3(-Q_2)}\Big|_{Q_2\rightarrow Q^{-1/2}\cdot s_{[r]}}
		=(s_{\pi_{i}}^{-1}\xi_t^{l_{i}-l_{j}})^{1/2}
		\cdot \frac{\vartheta(s_{[r]}s_{\pi_{i}}^{-1}\xi_t^{l_{i}-l_{j}})}{\vartheta(s_{[r]})}.
	\end{align*}
	As a result,
	the equation \eqref{eqn:Ft closed formula-Q_2} follows from setting $Q_2$ to be $Q^{-1/2}\cdot s_{[r]}$ in equation \eqref{eqn:Ft closed formula} and basic properties of matrix determinant.
\end{proof}

\section{Applications}
\label{sec:app}

\subsection{The quasimodularity for correlation functions of $t$-core partitions}
In this subsection,
we prove the quasimodularity for correlation functions of $t$-core partitions.
The main tool is the closed formula derived in the last section.

We refer the reader to the book \cite{DS05} by Diamond and Shurman for more details on Eisenstein series for congruence subgroups.
In this paper,
let $\Gamma_1(N)$ be the congruence subgroup of $SL_2(\mathbb{Z})$ such that
\begin{align*}
	\Gamma_1(N)
	=\left\{\left(\begin{matrix}
		a\ b\\
		c\ d
	\end{matrix}\right)
	\equiv \left(\begin{matrix}
		1\ \star\\
		0\ 1
	\end{matrix}\right) \mod N \right\}
	\subseteq SL_2(\mathbb{Z}).
\end{align*}
Let $s_j=e^{z_j}, j=1,\dots,n$.
We expand the $n$-point function $F_t(Q;s_1,\dots,s_n)$ of the $t$-core partitions by
\begin{align}
	F_t(Q;s_1,\dots,s_n)
	=\sum_{l_1,l_2,\dots,l_n=0}^\infty
	\langle f_{l_1}f_{l_2}\dots f_{l_n}\rangle_Q^{\text{t-core}}
	\cdot \prod_{j=1}^n z_j^{l_j-1}.
\end{align}
The $f_k(\cdot)$ could be regarded as a function on the set of partitions,
which is exactly equal to the function $Q_k(\cdot)$ used by Zagier (see Section 3 in \cite{Z16}).
The $\langle \cdots \rangle_Q^{\text{t-core}}$ is similar to the Bloch--Okounkov's $q$-bracket but restricted to the set of $t$-core partitions.
Then the correlation functions of $t$-core partitions
$$\langle f_{l_1}f_{l_2}\cdots f_{l_n}\rangle_Q^{\text{t-core}}$$
are functions of $Q=e^{2\pi i \tau}$.
\begin{prop}\label{prop:quasimod}
	For any given non-negative integers $l_1,\dots,l_n$,
	the correlation function of $t$-core partitions $\langle f_{l_1}f_{l_2}\dots f_{l_n}\rangle_Q^{\textup{t-core}}$ is a quasimodular form of weight at most $\sum_{j=1}^n l_j$ for the congruence subgroup $\Gamma_1(t)$.
\end{prop}
\begin{proof}
To study the correlation function of the $t$-core partitions,
we need to use the closed formula for the $n$-point function $F_t(Q;s_1,\dots,s_n)$ derived in the last Section.
The case of $r=1$ of equation \eqref{eqn:Ft closed formula-Q_2} gives
\begin{align}
	\begin{split}
		F_t(Q;s_1,&\dots,s_n)
		=\frac{1}{\prod_{j=1}^n(s_j^{t/2}-s_j^{-t/2})
			\cdot \vartheta(s_{[n]\setminus\{1\}}^{-1})}
		\cdot\sum_{k=1}^{n}
		\bigg(\frac{1}{\vartheta(s_1)^{k-1}}\\
		&\cdot
		\sum_{\{\pi_1,\dots,\pi_k\}\vdash[n]}
		\sum_{l_1,\dots,l_k=1}^t
		\prod_{m=1}^k \frac{\prod_{a=1}^t\vartheta(s_{\pi_m}\xi_t^{a-l_m})}
		{\prod_{a=1 \atop a\neq l_m}^t\vartheta(\xi_t^{a-l_m})}
		\cdot \det\Big(\frac{\vartheta(s_{[1]}s_{\pi_{i}}^{-1}\xi_t^{l_{i}-l_{j}})}
		{\vartheta(s_{\pi_{i}}\xi_t^{l_{j}-l_{i}})}\Big)_{i,j=1}^{k}\bigg).
	\end{split}
\end{align}
We note that the right hand side of above equation is a linear combination of terms of the following forms
\begin{align}\label{eqn:as linear}
	\frac{1}{\prod_{j=1}^n(s_j^{t/2}-s_j^{-t/2})}
	\cdot \prod_{i}\frac{\vartheta(s_{C_{i}}s_{D_{i}}^{-1}\xi_t^{c_{i}})}{\vartheta(s_{C'_{i}}s_{D'_{i}}^{-1}\xi_t^{c_{i}})},
\end{align}
where $C_i, C'_i, D_{i}, D'_{i}$ represent some subsets of $[n]=\{1,\dots,n\}$ and $c_i$ could be $0,\dots,t-1$.
The second factor in above equation \eqref{eqn:as linear} is known as a theta ratio in \cite{Eng21,EO06},
whose expansion yields quasimodular forms.
For example,
see Section 5.3 in \cite{Eng21} for the following equation
\begin{align*}
	\log \frac{\vartheta(\xi_t^r e^{z})} {\vartheta(\xi_t^r)}
	=\sum_{l=1}^\infty \frac{z^l}{l!}E^{r}_l(Q),
\end{align*}
where $E^{r}_l(Q)$ are quasimodular forms of weight $l$ for the congruence subgroup $\Gamma_1(t)$.
On the other hand,
since $s_j=e^{z_j}$,
the first factor in above equation \eqref{eqn:as linear} can be expanded in terms of
\begin{align*}
	\frac{1}{s_j^{t/2}-s_j^{-t/2}}
	=\frac{1}{2\sinh(tz_j/2)}
	=\frac{1}{t}z_j^{-1}
	-\frac{1}{24}tz_j
	+\frac{7}{5760}t^3z_j^3
	-\frac{31}{967680} t^5 z_j^5
	+\cdots.
\end{align*}
As a consequence,
the coefficient of $s_1^{l_1-1}\cdots s_{n}^{l_n-1}$ in the $n$-point function $F_t(Q;s_1,\dots,s_n)$ of $t$-core partitions must be a linear combination of products of Eisenstein series $E^{r}_l(Q)$,
whose total weights are at most $\sum_{j=1}^n l_j$.
This finishes the proof of this proposition.
\end{proof}

\subsection{The closed formulas for the $n=1,2,3$ cases}
We write down the explicit formulas for the $n$-point function of the $t$-core partitions when $n=1,2,3$.
\begin{ex}
	When $n=1$,
	the one-point function $F_t(Q;s)$ of $t$-core partitions is given by
	\begin{align}
		F_t(Q;s)=
		\frac{t}{(s^{t/2}-s^{-t/2})}
		\cdot \prod_{a=1}^{t-1} \frac{\vartheta(s\xi_t^{a};Q)}
		{\vartheta(\xi_t^{a};Q)}.
	\end{align}
\end{ex}

\begin{ex}
	When $n=2$,
	the two-point function $F_t(Q;s_1,s_2)$ of $t$-core partitions is given by
	\begin{align}\label{eqn:Fn=2 with Q2}
		\begin{split}
			F_t(Q;s_1,s_2)
			=&\frac{1}{\prod_{j=1}^2(s_j^{t/2}-s_j^{-t/2})}
			\cdot \Bigg(t\prod_{a=1}^{t-1} \frac{\vartheta(s_1s_2\xi_t^{a})}{\vartheta(\xi_t^{a})}
			+\frac{\vartheta(s_1)\vartheta(s_2)}{\Theta_3(-Q_2s_1^{-1}s_2^{-1})\Theta_3(-Q_2)}\\
			&\quad\cdot \prod_{a=1}^{t-1} \frac{\vartheta(s_1\xi_t^{a})\vartheta(s_2\xi_t^{a})}
			{\vartheta(\xi_t^{a})^2}
			\cdot \sum_{l_1,l_2=1\atop l_1\neq l_2}^t
			\left|\begin{matrix}
				\frac{\Theta_3(-Q_2s_1^{-1})}
				{\vartheta(s_1)}
				& \frac{\Theta_3(-Q_2s_1^{-1}\xi_t^{l_{1}-l_{2}})}
				{\vartheta(s_1\xi_t^{l_{2}-l_{1}})}\\
				\frac{\Theta_3(-Q_2s_2^{-1}\xi_t^{l_{2}-l_{1}})}
				{\vartheta(s_2\xi_t^{l_{1}-l_{2}})}
				& \frac{\Theta_3(-Q_2s_2^{-1})}
				{\vartheta(s_2)}
			\end{matrix}\right|\Bigg).
		\end{split}
	\end{align}
\end{ex}

As what we have done in Corollary \ref{cor:Ft closed formula-Q2},
we can set $Q_2$ in equation \eqref{eqn:Fn=2 with Q2} to be an arbitrary value to obtain a simpler formula for $F_2(Q;s_1,s_2)$.
For example,
when setting $Q_2=Q^{-1/2}s_1$,
we have
\begin{align}\label{eqn:Fn=2 without Q2}
	\begin{split}
		F_t(Q;s_1,s_2)
		=&\frac{1}{\prod_{j=1}^2(s_j^{t/2}-s_j^{-t/2})}
		\cdot \Bigg(t\prod_{a=1}^{t-1} \frac{\vartheta(s_1s_2\xi_t^{a})}{\vartheta(\xi_t^{a})}
		\\
		&\qquad\qquad+ \prod_{a=1}^{t-1} \frac{\vartheta(s_1\xi_t^{a})\vartheta(s_2\xi_t^{a})}
		{\vartheta(\xi_t^{a})^2}
		\cdot\sum_{l_1,l_2=1\atop l_1\neq l_2}^t
		\frac{\vartheta(s_1^{-1}s_2\xi_t^{l_{1}-l_{2}})\vartheta(\xi_t^{l_{1}-l_{2}})}
		{\vartheta(s_1^{-1}\xi_t^{l_{1}-l_{2}})\vartheta(s_2\xi_t^{l_{1}-l_{2}})}\Bigg).
	\end{split}
\end{align}
It will be interesting to provide a direct combinatorial  proof of the equivalence between the above two formulas \eqref{eqn:Fn=2 with Q2} and \eqref{eqn:Fn=2 without Q2}.
It seems that this equality should be a corollary of the Frobenius determinantal formula (see \cite{Fro} and Lemma 4.3 in \cite{Rains}).

\begin{ex}
	When $n=3$,
	the three-point function $F_t(Q;s_1,s_2,s_3)$ of $t$-core partitions is given by
	\begin{align*}
		F_t(Q;&s_1,s_2,s_3)
		=\frac{1}{\prod_{j=1}^3(s_j^{t/2}-s_j^{-t/2})}
		\cdot\Bigg(t\prod_{a=1}^{t-1} \frac{\vartheta(s_1s_2s_3\xi_t^{a})}
		{\vartheta(\xi_t^{a})}
		+\frac{1}{\Theta_3(-Q_2s_1^{-1}s_2^{-1}s_3^{-1})\Theta_3(-Q_2)}\\
		&\ \cdot\bigg(\vartheta(s_1)\vartheta(s_2s_3)
		\cdot\prod_{a=1}^{t-1} \frac{\vartheta(s_1\xi_t^{a})\vartheta(s_2s_3\xi_t^{a})}
		{\vartheta(\xi_t^{a})^2}
		\cdot\sum_{l_1,l_2=1\atop l_1\neq l_2}^t
		\left|\begin{matrix}
			\frac{\Theta_3(-Q_2s_1^{-1})}
			{\vartheta(s_1)}
			& \frac{\Theta_3(-Q_2s_1^{-1}\xi_t^{l_{1}-l_{2}})}
			{\vartheta(s_1\xi_t^{l_{2}-l_{1}})}\\
			\frac{\Theta_3(-Q_2s_2^{-1}s_3^{-1}\xi_t^{l_{2}-l_{1}})}
			{\vartheta(s_2s_3\xi_t^{l_{1}-l_{2}})}
			& \frac{\Theta_3(-Q_2s_2^{-1}s_3^{-1})}
			{\vartheta(s_2s_3)}
		\end{matrix}\right|\\
		&\ \ \ +\vartheta(s_1s_2)\vartheta(s_3)
		\cdot\prod_{a=1}^{t-1} \frac{\vartheta(s_1s_2\xi_t^{a})\vartheta(s_3\xi_t^{a})}
		{\vartheta(\xi_t^{a})^2}
		\cdot\sum_{l_1,l_2=1\atop l_1\neq l_2}^t
		\left|\begin{matrix}
			\frac{\Theta_3(-Q_2s_1^{-1}s_2^{-1})}
			{\vartheta(s_1s_2)}
			& \frac{\Theta_3(-Q_2s_1^{-1}s_2^{-1}\xi_t^{l_{1}-l_{2}})}
			{\vartheta(s_1s_2\xi_t^{l_{2}-l_{1}})}\\
			\frac{\Theta_3(-Q_2s_3^{-1}\xi_t^{l_{2}-l_{1}})}
			{\vartheta(s_3\xi_t^{l_{1}-l_{2}})}
			& \frac{\Theta_3(-Q_2s_3^{-1})}
			{\vartheta(s_3)}
		\end{matrix}\right|\\
		&\ \ \ +\vartheta(s_1s_3)\vartheta(s_2)
		\cdot\prod_{a=1}^{t-1} \frac{\vartheta(s_1s_3\xi_t^{a})\vartheta(s_2\xi_t^{a})}
		{\vartheta(\xi_t^{a})^2}
		\cdot\sum_{l_1,l_2=1\atop l_1\neq l_2}^t
		\left|\begin{matrix}
			\frac{\Theta_3(-Q_2s_1^{-1}s_3^{-1})}
			{\vartheta(s_1s_3)}
			& \frac{\Theta_3(-Q_2s_1^{-1}s_3^{-1}\xi_t^{l_{1}-l_{2}})}
			{\vartheta(s_1s_3\xi_t^{l_{2}-l_{1}})}\\
			\frac{\Theta_3(-Q_2s_2^{-1}\xi_t^{l_{2}-l_{1}})}
			{\vartheta(s_2\xi_t^{l_{1}-l_{2}})}
			& \frac{\Theta_3(-Q_2s_2^{-1})}
			{\vartheta(s_2)}
		\end{matrix}\right|\bigg)\\
		&+\frac{\vartheta(s_1)\vartheta(s_2)\vartheta(s_3)}{\Theta_3(-Q_2s_1^{-1}s_2^{-1}s_3^{-1})\Theta_3(-Q_2)^2}
		\cdot \prod_{j=1}^3\prod_{a=1}^{t-1} \frac{\vartheta(s_j\xi_t^{a})}{\vartheta(\xi_t^{a})}\\
		&\qquad\qquad\cdot\sum_{l_1,l_2,l_3=1\atop l_1\neq l_2\neq l_3\neq l_1}^t
		\left|
		\begin{matrix}
			\frac{\Theta_3\big(-Q_2s_1^{-1};Q\big)}
			{\vartheta(s_1;Q)}
			&\frac{\Theta_3\big(-Q_2s_1^{-1}\xi_t^{l_1-l_2};Q\big)}
			{\vartheta(s_1\xi_t^{l_2-l_1};Q)}
			&\frac{\Theta_3\big(-Q_2s_1^{-1}\xi_t^{l_1-l_3};Q\big)}
			{\vartheta(s_1\xi_t^{l_3-l_1};Q)}\\
			\frac{\Theta_3\big(-Q_2s_2^{-1}\xi_t^{l_2-l_1};Q\big)}
			{\vartheta(s_2\xi_t^{l_1-l_2};Q)}
			&\frac{\Theta_3\big(-Q_2s_2^{-1};Q\big)}
			{\vartheta(s_2;Q)}
			&\frac{\Theta_3\big(-Q_2s_2^{-1}\xi_t^{l_2-l_3};Q\big)}
			{\vartheta(s_2\xi_t^{l_3-l_2};Q)}\\
			\frac{\Theta_3\big(-Q_2s_3^{-1}\xi_t^{l_3-l_1};Q\big)}
			{\vartheta(s_3\xi_t^{l_1-l_3};Q)}
			&\frac{\Theta_3\big(-Q_2s_3^{-1}\xi_t^{l_3-l_2};Q\big)}
			{\vartheta(s_3\xi_t^{l_2-l_3};Q)}
			&\frac{\Theta_3\big(-Q_2s_3^{-1};Q\big)}
			{\vartheta(s_3;Q)}\\
		\end{matrix}\right|\Bigg).
	\end{align*}
\end{ex}

\subsection{The case of 2-core partitions}

All the 2-core partitions are given by
$$\nu^n=(n,n-1,\dots,1),\ \ n\in\mathbb{N}.$$
Then from the definition of $n$-point function $F_2(Q;s_1,\cdots,s_n)$ of 2-core partitions,
we have
\begin{align}
	\label{eqn:t=2 trueF2}
	\begin{split}
		F_2(Q;s_1,\cdots,s_n)
		=&\frac{1}{\sum_{n=0}^\infty Q^{n(n+1)/2}}
		\cdot \sum_{n=0}^\infty
		\bigg(\prod_{j=1}^n \sum_{i=1}^\infty s_j^{\nu^n_i-i+\half}
		\cdot Q^{n(n+1)/2}\bigg)\\
		=&\frac{1}{\sum_{n=0}^\infty Q^{n(n+1)/2}}
		\cdot \sum_{n=0}^\infty
		\bigg(\prod_{j=1}^n \frac{s_j^{n+1/2}+s_j^{-n-1/2}}{s_j-s_j^{-1}} \cdot Q^{n(n+1)/2}\bigg).
	\end{split}
\end{align}
For the case of $t=2, n=1$,
Theorem \ref{thm:main formula} gives
\begin{align}\label{eqn:t=2n=1}
	F_2(s)=
	\frac{2}{(s-s^{-1})}
	\cdot \frac{\vartheta(-s;Q)}
	{\vartheta(-1;Q)},
\end{align}
whose leading terms are given by
\begin{align*}
	F_2(Q;s)
	=&\frac{\sqrt{s}}{s-1}+\frac{s-1}{\sqrt{s}} Q-\frac{s-1}{\sqrt{s}} Q^{2}+\frac{(s-1) (s+1)^{2}}{s^{\frac{3}{2}}} Q^{3}-\frac{(s-1) (s^{2}+3 s+1)}{s^{\frac{3}{2}}} Q^{4}\\
	&+\frac{(s-1) (s^{2}+4 s+1)}{s^{\frac{3}{2}}} Q^{5}+\frac{(s-1) (s^{2}-3 s+1) (s+1)^{2}}{s^{\frac{5}{2}}} Q^{6}+\mathrm{O}(Q^{7}).
\end{align*}
Denote by $E_{2k}(Q)$ the classical Eisenstein series of weight $2k$ for the full modular group $SL_2(\mathbb{Z})$ with $k\geq1$,
i.e.,
\begin{align*}
	E_{2k}(Q)
	=\frac{\zeta(1-2k)}{2}
	+\sum_{n=1}^\infty \sum_{d|n \atop d>0}d^{2k-1} Q^n,
\end{align*}
where $Q=e^{2\pi i \tau}$.
Then under $s=e^z$, equation (37) in \cite{EO06} gives
\begin{align}\label{eqn:log theta as E t=2}
	\log \frac{\vartheta(-s;Q)}
	{\vartheta(-1;Q)}
	=2\sum_{k=1}^\infty
	\frac{z^{2k}}{(2k)!}\Big(E_{2k}(Q)-2^{2k}E_{2k}(Q^2)\Big).
\end{align}

Combining equations \eqref{eqn:t=2 trueF2}, \eqref{eqn:t=2n=1} and \eqref{eqn:log theta as E t=2},
we obtain
\begin{align*}
	\frac{\sum_{n=0}^\infty
		(s^{n+1/2}+s^{-n-1/2}) \cdot Q^{n(n+1)/2}}
	{2\sum_{n=0}^\infty Q^{n(n+1)/2}}
	=\exp\Big(2\sum_{k=1}^\infty
	\frac{z^{2k}}{(2k)!}\big(E_{2k}(Q)-2^{2k}E_{2k}(Q^2)\big)\Big).
\end{align*}

\subsection{The $n$-point function of all integer partitions: Bloch--Okounkov's case}
Notice that the $q$-deformed $n$-point function introduced in this paper also generalizes the ordinary $n$-point function of all integer partitions studied by Bloch--Okounkov (see Proposition \ref{prop:qZn lim F}).
Then based on the formula for the $q$-deformed $n$-point function performed in Proposition \ref{thm:qZn as det},
we can also derive the following formula for the ordinary $n$-point function of all integer partitions.
\begin{prop}
	Denote by $F(Q;s_1,\dots,s_n)$ the $n$-point function of all integer partitions.
	Then we have the following formula
	\begin{align}
		F(Q;s_1,\dots,s_n)
		=\frac{1}{\Theta_3(Q_2)^{n-1}\cdot \Theta_3(Q_2 s_{[n]}^{-1})}
		\cdot[w_1^0\dots w_n^0]\det \Big(\frac{\Theta_3(Q_2s_i^{-1}w_i^{-1}w_j)}{\vartheta(s_iw_iw_j^{-1})}\Big).
	\end{align}
\end{prop}
\begin{proof}
	Based on Propositions \ref{prop:qZn lim F} and \ref{thm:qZn as det},
	we need to compute
	\begin{align*}
		\lim_{q\rightarrow \infty}
		\frac{E(Q;Q_1,q;s_iw_i,w_j)}{1-s_i^{-1}w_i^{-1}w_j}|_{Q_1=1}
		=&\frac{1}{1-s_i^{-1}w_i^{-1}w_j}
		\frac{\big(Q;Q\big)_{\infty}^2}
		{\big(Qs_i^{-1}w_i^{-1}w_j,Qs_iw_iw_j^{-1};Q\big)_{\infty}}\\
		=&\frac{1}{\sqrt{s_i^{-1}w_i^{-1}w_j} \cdot \vartheta(s_iw_iw_j^{-1};Q)}.
	\end{align*}
	Then from the formula \eqref{eqn:qZn as det} for the $q$-deformed $n$-point function,
	we have
	\begin{align*}
		F(Q;s_1,\dots,s_n)
		=&\lim_{q\rightarrow \infty} Z(Q;Q_1,q;s_1,\dots,s_n)|_{Q_1=1}\\
		=&\frac{1}{\Theta_3(Q_2)^{n-1}\cdot \Theta_3(Q_2 s_{[n]}^{-1})}
		\cdot[w_1^0\dots w_n^0]\det \Big(\frac{\Theta_3(Q_2s_i^{-1}w_i^{-1}w_j)}{\vartheta(s_iw_iw_j^{-1})}\Big)_{i,j=1}^n.
	\end{align*}
\end{proof}

Based on the same method used in subsection \ref{sec:closed formula},
one may derive a new family of closed formulas for the Bloch--Okounkov's $n$-point function $F(Q;s_1,\dots,s_n)$, parameterized by $Q_2$.
It would be interesting to write it down explicitly and compare it with the Bloch--Okounkov's formula \cite{BO} (see also \cite{Z16,Zhou23}).

\section*{Conflict of interest and data availability statement}
The author states that there is no conflict of interest, and
no datasets were generated or analysed during the current study.

\vspace{.2in}
{\em Acknowledgements}.
The author would like to thank Zhiyuan Wang and Zhiyong Wang for their collaboration on related works.
This work was partially supported by the NSFC grant (No. 12401079).

\renewcommand{\refname}{Reference}
\bibliographystyle{plain}
\bibliography{reference}
\vspace{20pt} \noindent
\end{document}